\documentclass[a4paper,11pt,final]{article}

\usepackage{lastpage}
\usepackage[cm]{fullpage}
\usepackage{cite}

\usepackage{fancyhdr} 
\newcommand{\muSPARK}{$\mu$SPARK}

\usepackage{graphicx} 
\usepackage{multicol} 
\usepackage{hyperref} 
\usepackage{amssymb} 
\usepackage{amsmath} 
\usepackage{amsthm} 
\usepackage{mathtools} 
\usepackage{tikz-cd} 
\usepackage{wrapfig} 
\usepackage{subcaption} 
 
\usepackage{tikz}
\usetikzlibrary{positioning,calc}
 
\usepackage{listings}

\lstdefinestyle{spark}{
	language=Ada,
	keywordstyle=\bfseries\ttfamily\color[rgb]{0,0,1},
	identifierstyle=\ttfamily,
	commentstyle=\ttfamily\color[rgb]{0.133,0.545,0.133},
	stringstyle=\ttfamily\color[rgb]{0.627,0.126,0.941},
	  morekeywords=[1]some,
	showstringspaces=false,
	basicstyle=\small\ttfamily,
	stepnumber=1,
	numbersep=10pt,
	tabsize=2,
	breaklines=true,
	prebreak = \raisebox{0ex}[0ex][0ex]{\ensuremath{\hookleftarrow}},
	breakatwhitespace=false,
	aboveskip={1.5\baselineskip},
	columns=fixed,
	extendedchars=true,
        literate={>=}{>{}={}}{2}{<=}{<{}={}}{2}
}

\usepackage{xcolor}
\usepackage{rotating}

\usepackage{colortbl}

\usepackage{wasysym}

\newcommand{\valid}[1]{\cellcolor{green!13}#1}

\usepackage{bcprules} 
\usepackage{mathpartir} 

\newcommand{\reportsubject}{Master 2 Internship Report - MPRI}
\newcommand{\reporttitle}{Safe Pointers in SPARK 2014.}     
\newcommand{\reportauthor}{Georges-Axel \textsc{Jaloyan}} 
\newcommand{\HRule}{\rule{\linewidth}{0.5mm}}
\newcommand{\rulesep}{\unskip\ \vrule height -1ex\ }
\newtheorem{lemma}{Lemma}
\newtheorem{theorem}{Theorem}
\newtheorem{XxmpX}{Proof of theorem} 
\newenvironment{preuve}    
{\pushQED{\qed}\begin{XxmpX}}
	{\popQED\end{XxmpX}}

\begin{document}
	\begin{titlepage}
		
		\begin{center}

			\textsc{\Large \reportsubject\\[0.5cm]20 March 2017 - 4 August 2017}\\[2.5cm]
			\HRule \\[0.4cm]
			{\huge \bfseries \reporttitle}\\[0.4cm]
			\HRule \\[2cm]
			
			\begin{minipage}[t]{0.3\textwidth}
				\begin{flushleft} \large
					\emph{Author:}\\
					\reportauthor
				\end{flushleft}
			\end{minipage}
			\begin{minipage}[t]{0.6\textwidth}
				\begin{flushright} \large
					\emph{Supervisor:} \\
					Dr.~Yannick \textsc{Moy}
				\end{flushright}
			\end{minipage}
			\\[4cm]
			\centering
			\begin{minipage}[t]{0.4\textwidth}\vspace{40pt} 
				\begin{flushleft}
					\includegraphics[width=60mm]{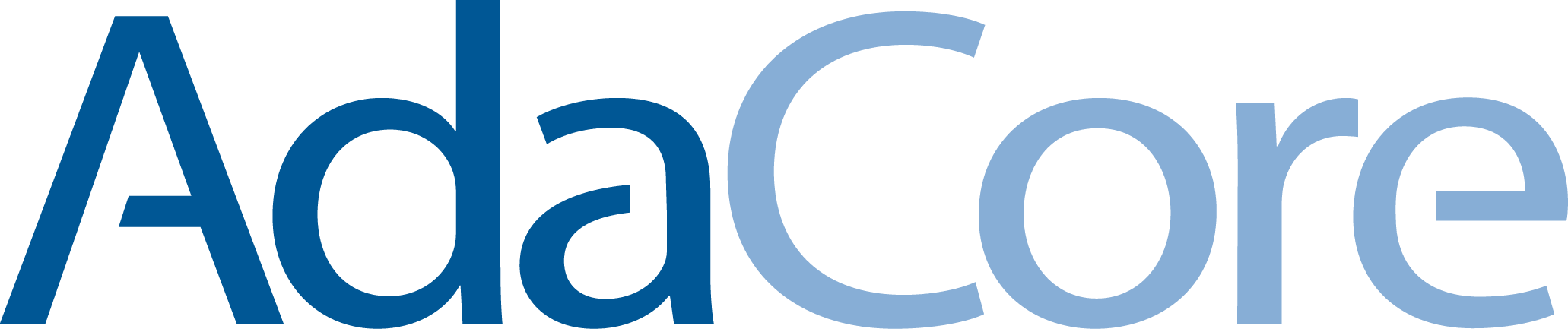}
				\end{flushleft}
			\end{minipage}
			\begin{minipage}[t]{0.4\textwidth}\vspace{0pt} 
				\begin{flushleft}
					\includegraphics[width=40mm]{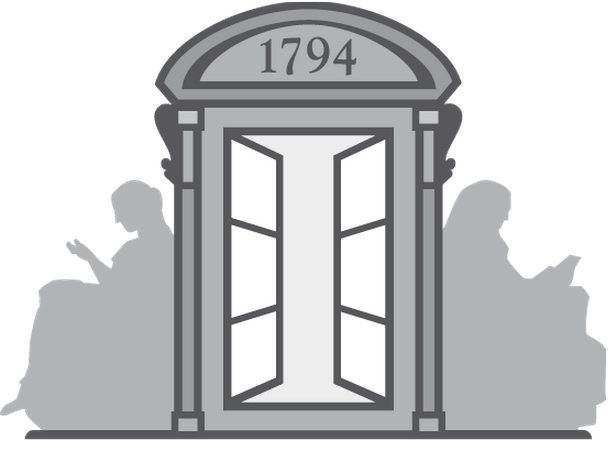}
				\end{flushleft}
			\end{minipage} 
		
			\begin{minipage}[t]{0.4\textwidth}\vspace{0pt} 
				\begin{flushleft}
					\textsc{\LARGE AdaCore}
				\end{flushleft}
			\end{minipage}
			\begin{minipage}[t]{0.4\textwidth}\vspace{0pt} 
				\begin{flushleft}
					\textsc{\LARGE \'Ecole normale sup\'erieure}
				\end{flushleft}
			\end{minipage} \\[1.5cm]
			\vfill
			
			{\large 21 August 2017}
			
		\end{center}
		
	\end{titlepage}

	\title{\reportsubject\\\reporttitle}
	
	\author{\reportauthor, supervised by Yannick \textsc{Moy}, AdaCore.}
	
	\date{20 March 2017 - 4 August 2017}
	
	\maketitle
	
	\pagestyle{empty} %
	\thispagestyle{empty}

	\subsection*{The general context}
	This report is referring to an internship from the 20\textsuperscript{th} of March 2017 to the 4\textsuperscript{th} of August 2017 at \emph{AdaCore} located in Paris, France. AdaCore is a company providing commercial software solutions for Ada, a language targeting safety-critical applications. More specifically, SPARK Pro \cite{mccormick_chapin_2015, sparkrm, sparkuser} is a static analysis tool-suite using formal methods to deductively verify a subset of Ada (called SPARK), commercialized in partnership with Altran. The tool can be used to meet DO-178B/C (and the Formal Methods supplement DO-333), CENELEC 50128, IEC 61508, and DEFSTAN 00-56 standards. Its latest revision is called SPARK 2014.
	
	Deductive verification uses SMT solvers (CVC4, Z3, Alt-Ergo) and interactive theorem provers (Coq and Isabelle) to check the validity of logical propositions (called \emph{verification conditions}) generated from the source code and the specification. Those proofs are enough to give a strong guarantee that the program meets its specification.
	
	The tool is mainly based on Why3\cite{Why3}, an auto-active verification platform using an intermediate language, WhyML, that uses deductive verification techniques to prove user-specified properties (expressed as contracts in SPARK code) or other automatically generated properties such as \emph{absence of run-time errors} (AoRTE). The main limitation in SPARK 2014 is the absence of Ada pointers (called \emph{access types} in Ada\footnote{A more precise definition is given in Ada Reference Manual, section 3.10 \cite{adarm}.}), which may require some complex workarounds and code refactoring to make them accepted by SPARK Pro.
	
	\subsection*{The research problem}
	
	The main problem in adding access types to SPARK is to control aliasing (when multiple pointers point to the same object in memory). This heavily complicates proofs, and generally involves human-written annotations.\footnote{This is the case for example with separation logic\cite{seplogic}.} Cyclone and Rust are examples of languages that control aliasing using compile-time restrictions but none of them allows to prove programs. 
	
	SPARK already allows to control aliasing between "references" (implicit pointers created by the compiler for parameter passing), but not arbitrary pointers. To the best of our knowledge, the proposed method in this report is a novel approach for controlling aliasing introduced by arbitrary pointers in a programming language supported by proof. Previous attempts have relied on a combination of access policy enforced by proof mechanisms (like the so-called Boogie methodology in Spec\#). Our approach does not require user annotations or proof of verification conditions, which makes it much simpler to adopt.
	
	The aim of the internship is to design and implement an aliasing control mechanism that would allow using pointers in SPARK. Such a design needs to be (almost) fully automatic, to do not impact existing proofs of programs without pointers, and should be expressive enough so that existing code (such as the libraries of drivers and containers developed at AdaCore) could be verified using SPARK. Some mathematical guarantees have to be given on the new design, and its implementation should be tested against a large existing code-base in SPARK and Ada.
	
	\subsection*{Your contribution}

	The main scientific contributions of this report are the following: the design of rules for access types in SPARK, inspired by Rust's \emph{borrow-checker} and \emph{affine typing}, enforcing \emph{Concurrent Read, Exclusive Write} (\emph{CREW}) and an intra-procedural move semantics. We also present an implementation in Ada's compiler (called \emph{GNAT}) front-end as a static analysis pass. Then we formalize these rules on a subset of SPARK and give a proof of non-aliasing, as done for a subset of Rust in \cite{rustbelt} and \cite{patina}. Finally, we study these new rules on existing code-base and examples, in order to check for their applicability for real software, and compare them to Rust on some idiomatic constructs.
	
	\subsection*{Arguments supporting its validity}
	
	The rules established by the author were validated at several levels. The Language Design Committee of Ada/SPARK, an international board of experts in Ada, validated the rules as a standard feature for a future version of SPARK (most likely SPARK Pro 19). Comparison with Rust showed a great improvement on aliasing control at the expense of not handling automatic reclamation and non-null pointer coercion, being both checked by other passes in SPARK Pro. 
	
	Tests on industrial Ada code containing pointers (Ada drivers library, big strings containers library) showed that this analysis only requires minor changes to the code in order for it to be accepted in SPARK Pro.
	
	\subsection*{Summary and future work}
	
	First future step would consist in extending SPARK tool-chain to have its flow analysis and proof mechanisms compatible with the new pointers added. Indeed, during the internship, only the anti-aliasing rules have been implemented, whereas the remaining part of the SPARK back-end needs modifications to analyze SPARK code containing pointers. For this purpose, a post-doc has been hired by AdaCore to work on those issues.
	
	Next steps would consist in finalizing the implementation on corner-cases of SPARK, and presenting the results at an international scientific conference. A more general adoption of those rules by different customers as well as the generalization of non-aliasing proofs to the complete SPARK language is expected. A good evolution could also be to formalize those proofs in an interactive theorem prover (such as Coq or Isabelle) in order to have stronger guarantees on the system. 
	
	This internship gives us a better understanding of borrow-checkers and the kind of guarantees they should give, not only theoretically (formalization of what is a borrow-checker, proof of soundness), but also in practice (what a borrow-checker should exactly take care of, how to make it understandable to users, how to integrate it with other tools for static analysis).
	
	A good next question is the study of the constructs that can be implemented using this borrow-checker. How to make the most of a borrow-checker in order to improve parallelism, compiler optimizations, automatic reclamation? Answering these questions could lead to useful applications when transferring the borrow-checker to full Ada, providing it as an additional feature of the language.
	
	\subsection*{Acknowledgment}
	
	This work is partly supported by the Joint Laboratory ProofInUse (ANR-13-LAB3-0007) and project VECOLIB (ANR-14-CE28-0018) of the French National Research Agency (ANR).
	
	A number of people have provided input and advice; we particularly thank
        Claire Dross (AdaCore), Tucker Taft (AdaCore), Raphaël Amiard (AdaCore),
        Steve Baird (AdaCore), Claude March\'e (Inria), Andrei Paskevich (LRI),
        Jean-Christophe Filliâtre (LRI), Jacques-Henri Jourdan (MPI-SWS),
        Sylvain Dailler (AdaCore), Cl\'ement Fumex (AdaCore).

	\cleardoublepage 
	{\tableofcontents 
	\sloppy          

	\cleardoublepage
	
	\section{Preliminaries}
	\label{Preliminaries}
	
	\subsection{Ada}
	
	Ada is a general purpose language originally standardized in 1983 for the Department of Defense as a safety-oriented language targeting safety-critical systems. Its latest revision, Ada 2012 is defined by ISO/IEC~8652:2012. It features strong typing, objects and packages for modularity, and concurrency handling (Ravenscar profile being the best known\cite{Ravenscar}). A typical Ada program is made of several packages that can contain declarations (functions, procedures, packages, variables, and types declarations) or initialization code, that will be executed when the package is loaded.
	
	A lot of features are unique to Ada amongst programming languages for safety-critical applications, in particular modes of parameters, that can be of mode \texttt{in}, \texttt{out} or \texttt{in-out}.\footnote{We added the dash to \texttt{in-out} for readability purposes.} An \texttt{in} parameter can be read but not modified by the callee whereas an \texttt{in-out} parameter can be read and modified, and \texttt{out} parameters can only be modified (and are considered uninitialized), their values being sent back to the caller. 
	
	The language has many features inherited from older versions, and its safety concerns require it to have a well-defined semantics detailed in the \emph{Ada Reference Manual} \cite{adarm}.
	
	\subsection{SPARK 2014}
	SPARK~2014 is a subset of Ada designed and suited for static analysis using formal methods. SPARK~2014 features specification constructs allowing writing pre and post conditions, assertions, invariants, data dependencies. A SPARK program is either made of  fully written SPARK code, or is a subset of an Ada program where some subprograms or packages have been marked with a special annotation called \texttt{SPARK\_Mode}, which allows mixing SPARK and non-SPARK code.
	
	The main restrictions with respect to full Ada are the absence of pointers and aliasing (some checks are already done in SPARK to prevent aliasing between parameters and globals) and concurrency handling (only \emph{extended Ravenscar profile} \cite{extendedRavenscar} is allowed). We also require that objects must fulfill the \emph{Liskov Substitution Principle} (LSP) \cite{Liskov}, which prevents a more specialized subclass to have its invariants violated by the methods of the parent class. More specifically this forces the more specialized subclass to implement sub-behaviors of the parent class with respect to contracts expressed on the parent methods. 
	
	In this report, we study a small subset of SPARK enriched with pointers, called {\muSPARK}, the grammar of which is given in Appendix~\ref{sec:Grammar}. Each {\muSPARK} program consists of a single file that contains a main procedure. Each procedure starts with a list of definitions (type, procedure or variable) followed by a body. This body is a sequence of statements, typically assignments, allocations, procedure calls, conditional statements. We abstract out the condition of the if statement, considering it as a non-deterministic branching.
	
	\subsection{Proofs in presence of aliasing}
	To illustrate how aliasing can cause problems with proofs, let us consider the following function with two parameters, that increases each of them. The contract that could be expected of such a function (written in C annotated with ACSL contracts \cite{acslrm}), ensures that the two parameters have their value increased by one at the end of the function.
	
\begin{lstlisting}[basicstyle=\small\ttfamily,basicstyle=\small\ttfamily,language=C]
/*@
 requires \valid(x) && \valid(y);
 ensures (*x == \old(*x)+1 && *y == \old(*y)+1 );
*/
void inc (int* x, int* y) {
  *x++; *y++;
}
\end{lstlisting}

	This contract is false in presence of aliasing. Indeed if \texttt{x} and \texttt{y} point to the same variable, then the final value of each parameter would get increased by two instead of one. In ACSL, it is possible to use a logic function to express separation of sets of pointers by adding annotations specifying that two pointers should be considered as non-aliased. The contract hence becomes as following, which gets easily proved using the WP or Jessie plug-in of Frama-C toolset
	
\begin{lstlisting}[basicstyle=\small\ttfamily,language=C]
/*@
 requires \valid(x) && \valid(y);
 requires \separated(x,y);
 ensures (*x == \old(*x)+1 && *y == \old(*y)+1 );
*/
\end{lstlisting}
	
	\subsection{Rust}
	
	Rust is a new programming language created by Mozilla Foundation \cite{rust} and intended for system programming. It focuses mainly on memory safety, with the help of a powerful mechanism called a \emph{borrow-checker}, that prevents non-safe aliasing through pointer analysis with techniques inspired by affine types (in which resources can be used at most once) \cite{GIRARD19871, Affinetypes}. 
	
	The borrow-checker analyzes the source code looking for two different constructs: moves and borrows. 
	
	Moves happen at the right-hand side of assignments as well as when passing parameters to functions that are not borrowed (see after). Resources that are moved are consumed (exactly as in affine types), which allows only a single path to read or modify resources at a time. 
	
	Borrows are temporary grants of a given resource to another path, that can be either 
	\emph{immutable} (read-only, using the symbol \texttt{\&}) or \emph{mutable} (read-write, using the symbol \texttt{\&mut}). The borrow checker verifies that when grants are given, the CREW principle is respected, raising errors at each violation. As an example, let us consider the following code which borrows mutably twice the same path, triggering an error.
	
\begin{lstlisting}[basicstyle=\small\ttfamily,language=C]
let mut r = String::from("hello");

let mut s = r;   // move r to s
let mut s2 = r;  // error: use of moved value: `r`
let r1 = &mut s; // mutable borrow of s
let r2 = &mut s; // error[E0499]: cannot borrow `s` as mutable more than once
\end{lstlisting}
	
	Rust's borrow-checker also checks the lifetimes of pointers to make sure that no dangling pointer is used, and to provide a compile-time garbage collector.

	\subsection{Main objective}
	
	The goal of this internship is to design a mechanism that controls aliasing in SPARK, based on enforcing \emph{CREW}. We use a model in which a syntactical element has (full or partial) ownership of all data that is accessible from it (e.g. a pointer owns the data pointed by it, ...). 
	
	The CREW mechanism must ensure that two aliased paths do not have full ownership of the underlying data. More specifically, either one path has the full ownership (read-write) and the others have no ownership at all, or all of them have at most partial ownership (read-only).
	
	\section{Access types in {\muSPARK}}
	
	Three mechanisms have been created to ensure aliasing control, named respectively move, borrow and observe. The first one consists in checking assignments, whereas the two others focus on subprogram calls. We present the aliasing rules for {\muSPARK}, while explaining the main differences when passing to complete SPARK in section \ref{sec:completeSPARK}.
	
	\subsection{Definitions}
	\label{sec:Rules}
	We define \emph{paths} as abstractions of left values (called \emph{names} in Ada). More precisely, a path is a name in which all array indices have been abstracted and all dereferences (or equivalently \emph{indirections}) are made explicit.\footnote{Ada has implicit pointer dereference.} In the case of {\muSPARK}, names and paths are the same, given that implicit dereference as well as arrays are excluded from the language. Note that in Ada, dereferences are written using the \texttt{.all} selector.
	
	By introducing the notion of prefixes and extensions of paths, we can use trees to represent a set of paths, as shown in Figure~\ref{fig:original}. This tree is defined as the following. Each node represents a path, with the root being the base identifier. Those nodes can be of three different types: either an Integer node (such as \texttt{My\_Var.all.z}), a Pointer node (such as \texttt{My\_Var} or \texttt{My\_Var.all.y}), or a Record node (such as \texttt{My\_Var.all}).
	
	\begin{wrapfigure}{r}{0.35\textwidth}
		\vspace{-1cm}
		\includegraphics[width=0.35\textwidth]{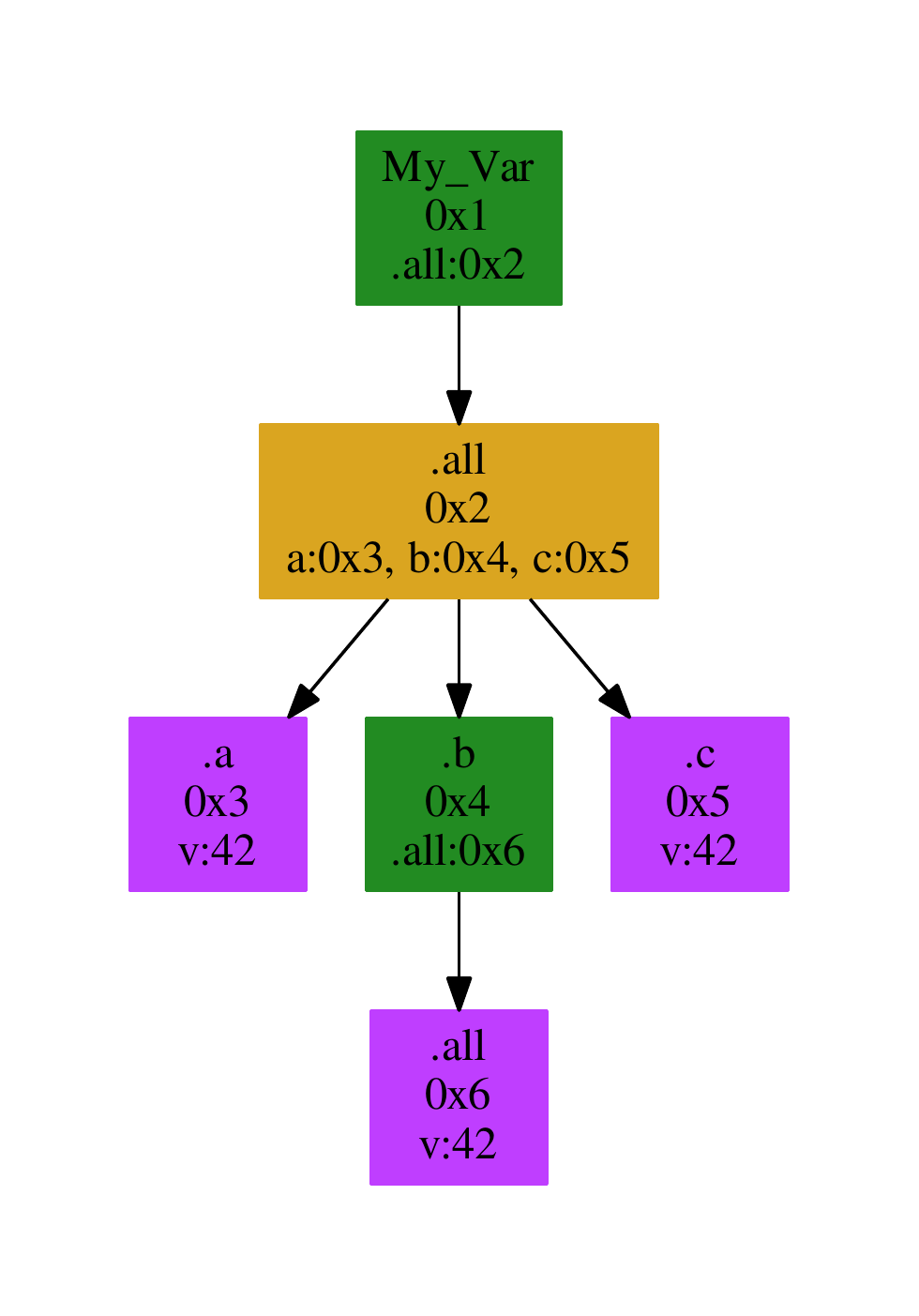}
		\vspace{-1cm}
		\caption{Example of a tree representing the following paths: \texttt{My\_Var}, \texttt{My\_Var.all}, \texttt{My\_Var.all.x}, \texttt{My\_Var.all.y}, \texttt{My\_Var.all.z}, \texttt{My\_Var.all.y.all}.}
		\label{fig:original}
		\vspace{-1cm}
	\end{wrapfigure}

	We will focus our analysis on types that can cause aliasing problems during assignment or parameter passing, that we call \emph{deep}, that are intuitively pointers and records having at least one deep type component. Formally, a type is said to be \emph{deep} if it is a type having an access part.\footnote{See Ada Reference Manual, section 3.2.(6/2).} Types that are not deep are called \emph{shallow}.

	\subsubsection{Reading and writing a path}
	Let us consider a path (example \texttt{My\_Var.all}). Our memory model gives to this path the ownership of all its extensions (hence all the subtree that is rooted at \texttt{My\_Var.all}). 
	
	When reading this path (for example at a the right-hand side of an assignment), every extension of this path is read (i.e. \texttt{My\_Var.all}, \texttt{My\_Var.all.x}, \texttt{My\_Var.all.y}, \texttt{My\_Var.all.y.all}, \texttt{My\_Var.all.z}). 
	
	Symmetrically, when writing this path (left-hand side of an assignment), only extensions that have the same number of dereferences are written (\texttt{My\_Var.all}, \texttt{My\_Var.all.x}, \texttt{My\_Var.all.y}, \texttt{My\_Var.all.z} but \textbf{not} \texttt{My\_Var.all.y.all}). Every pointer accessible from the path gets smashed (\texttt{My\_Var.all.y}), hence the memory areas designated by those pointers are replaced by new areas instead of being written.

	This asymmetry is fundamental when considering the rules for non-aliasing. Reading follows the pointers underlying a data structure whereas writing is stopped at those pointers, that get smashed instead.
	
	\subsubsection{Permissions}
	\label{sec:Permission}
	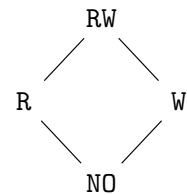
\begin{wrapfigure}{r}{0.30\textwidth}
		\centering
		\begin{tikzcd}[tips=false,column sep=1em,row sep=1.5em]
			&& \texttt{RW} \ar{dl} \ar{dr}\\
			& \texttt{R} \ar{dr} && \texttt{W} \ar{dl} \\
			&& \texttt{NO}
		\end{tikzcd}
		\caption{Hasse diagram for the permission lattice.}
		\label{fig:hasse}
		\vspace{-1cm}
	\end{wrapfigure}
	
	To explicitly control aliasing, we must restrict the readability and writability of some paths to ensure that only one path can write a given memory area at a time. For this, we add explicitly permissions on the previous trees, that can be of four different values: \texttt{RW} (read-write), \texttt{W} (write-only), \texttt{R} (read-only), \texttt{NO} (no permission). Those permissions form a lattice, whose Hasse diagram is given in Figure~\ref{fig:hasse}.
	
	Those trees enriched with permissions are called \emph{permission trees}, and their formal definition is given is Appendix~\ref{sec:Permission Rules for MuSPARK}. Note that those trees are of infinite depth when a structure has a component of type access to itself (linked lists, graphs,~...).
	
	\subsection{Moving}
	
	Moving is a mechanism that transfers ownership from one or more paths (called the \emph{moved paths}, the right-hand side) to another path (the \emph{assigned path}, the left-hand side). 
	
	In order to find all the paths present in the right-hand side that are considered as moved, we first define the concept of \emph{moved expression}, and then gather from those expressions every path on which we check the adequate permissions and update them accordingly. Finally, we update the permission of the assigned path.
	
	More formally, an expression appearing in the right-hand side of an assignment statement is a \emph{moved expression} when: 
	\begin{itemize}
		\item It is a deep type top-level expression (i.e. expressions that are not part of other expressions).
		\item It is the prefix of an \texttt{Access} attribute (equivalent to \texttt{\&} ``address-of" unary operator in C) of a moved expression.\footnote{This prefix is necessarily a name per Ada rules.}
		\item It is a deep type direct sub-expression of a moved expression.
	\end{itemize}

	When a moved expression is a name, then this yields directly a moved path. We process each moved path sequentially by checking first its Read-Write permission, and then deleting both read and write permission to any path that gets aliased with the assigned path. Using our tree representation, one could remark that when a path (ie. node) is aliased, then all the subtree rooted at that node is also aliased. Hence we can  directly consider aliases of trees.
	\begin{figure}[!hb]
		\centering
		\footnotesize
		\begin{subfigure}[t]{0.3\textwidth}
			\includegraphics[width=\textwidth]{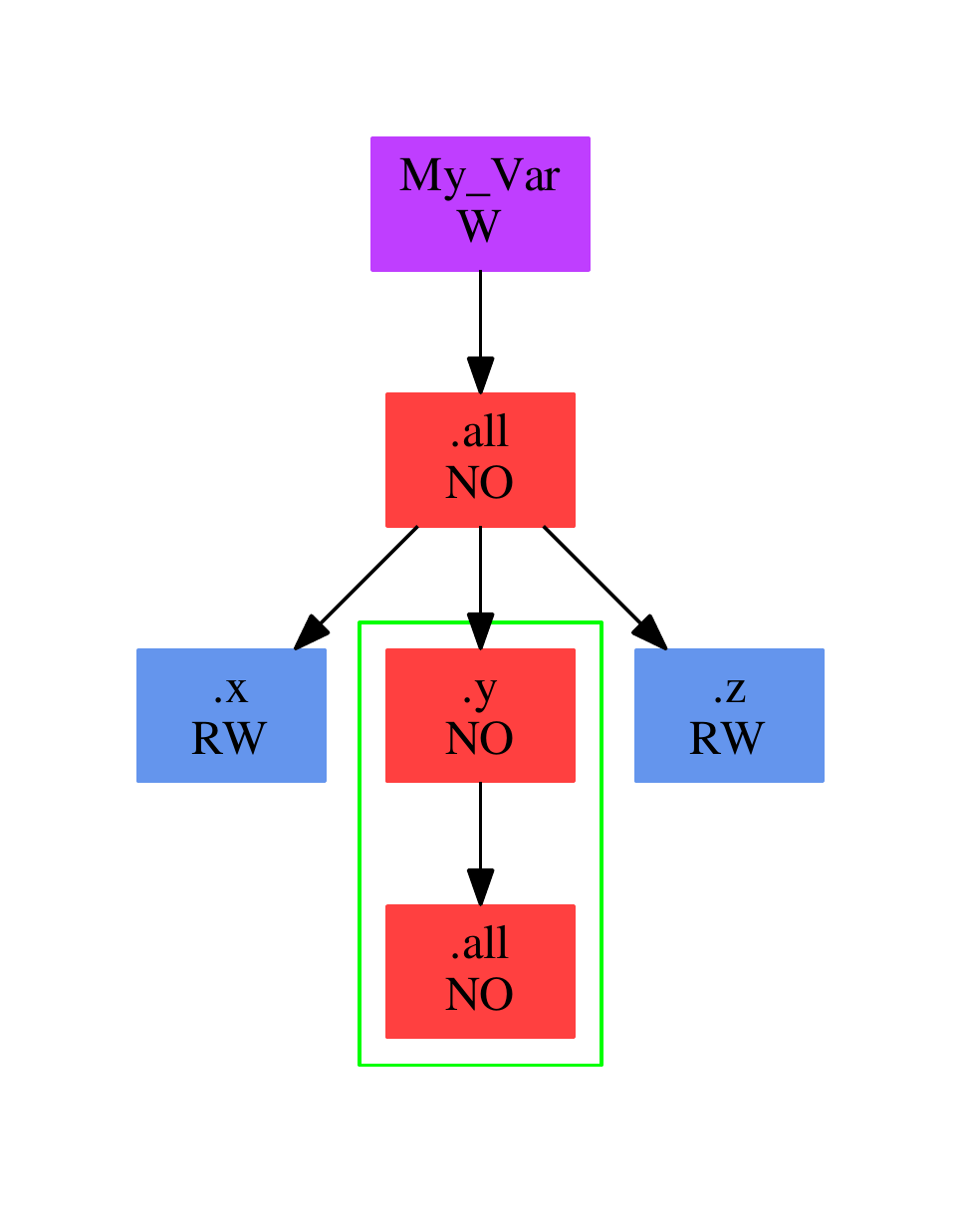}
			\centering
			\caption{Example of a tree with its permissions updated after the path \texttt{My\_Var.all.y} has been moved with \texttt{Access} attribute. The aliased subtree is framed in green.}
			\label{fig:myvarallyaccess}
		\end{subfigure}\rulesep
	\begin{subfigure}[t]{0.3\textwidth}
	\includegraphics[width=\textwidth]{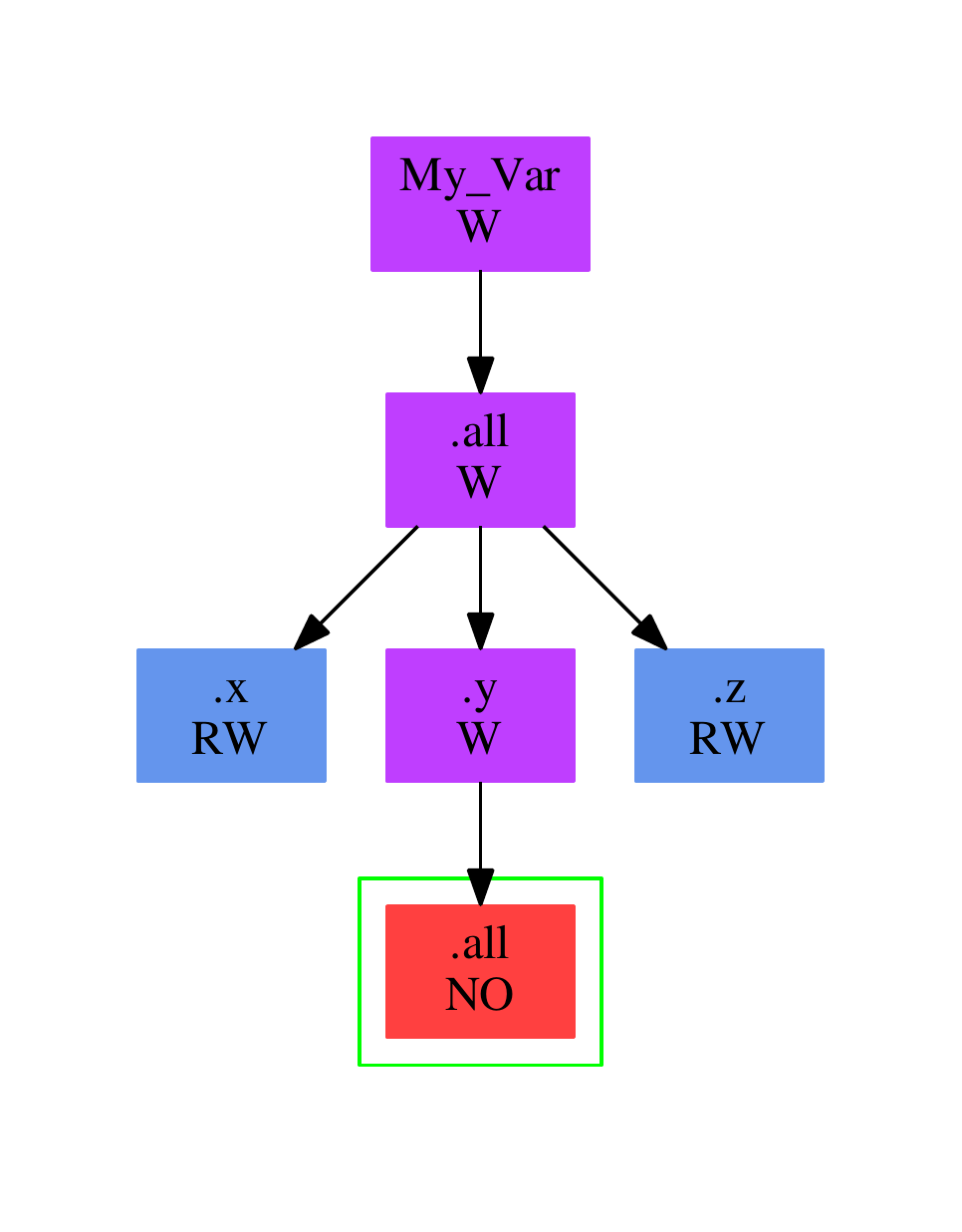}
	\centering
	\caption{Example of a tree with its permissions updated after the path \texttt{My\_Var.all.y} has been moved without \texttt{Access} attribute. The aliased subtree is framed in green.}
	\label{fig:myvarally}
\end{subfigure}\rulesep
\begin{subfigure}[t]{0.3\textwidth}
	\includegraphics[width=\textwidth]{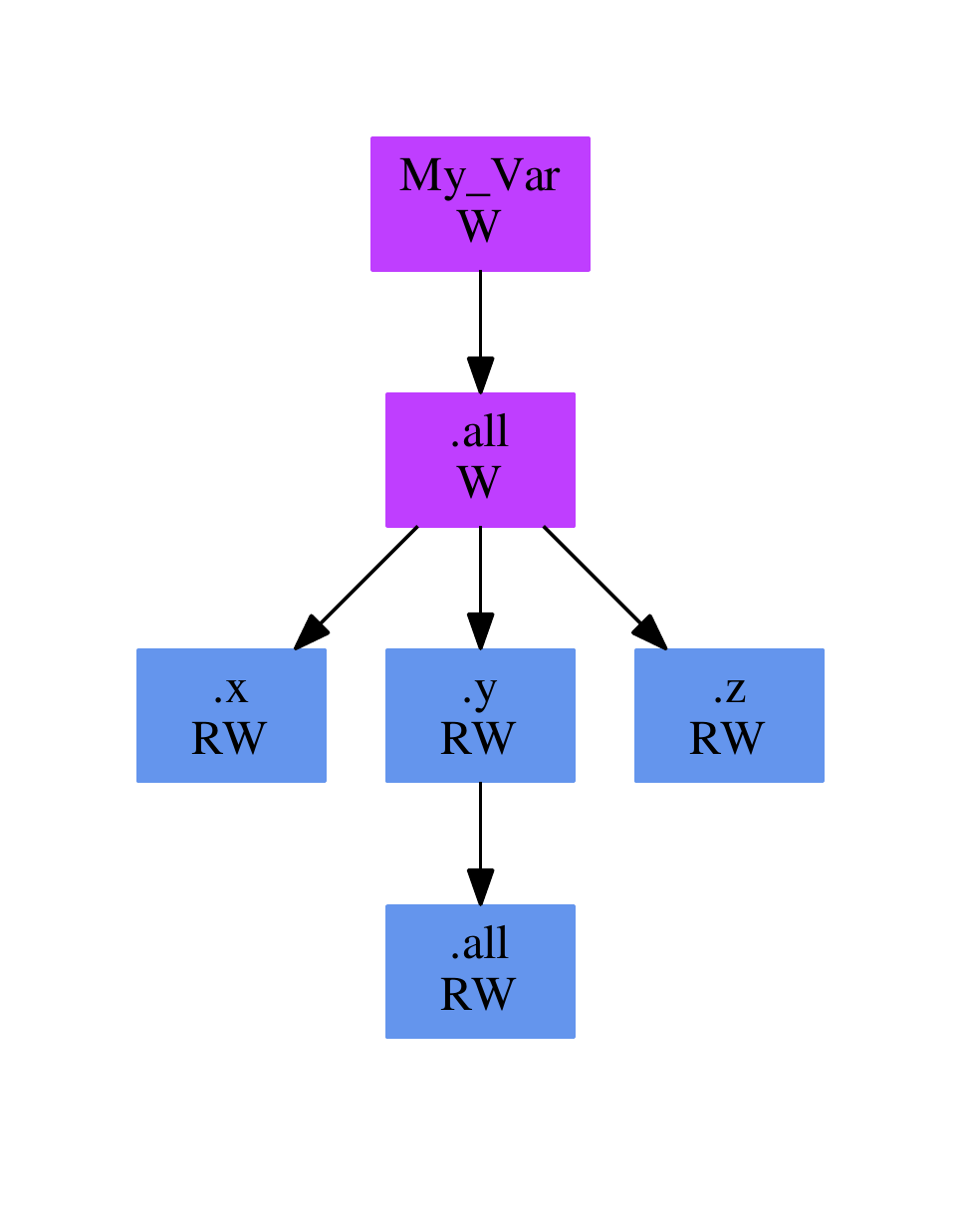}
	\centering
	\caption{The previous tree (Figure~\ref{fig:myvarally}) after assigning \texttt{My\_Var.all.y}. A further propagation will give \texttt{RW} permission to \texttt{My\_Var.all} and \texttt{My\_Var}.}
	\label{fig:assignmyvarally}
\end{subfigure}
	\end{figure}

	If the path is yielded from the prefix of an \texttt{Access} attribute, then the whole subtree is aliased. Hence, the whole subtree sees its permission changed to \texttt{NO}, as well as any ancestor node (prefix) of this path that have the same number of indirections. All other ancestors only lose their readability and get their permission updated to \texttt{W}.\footnote{The exact wording as appearing in the SPARK RM says: ``After moving a path under \texttt{‘Access} attribute, the permission for the path, any of its extensions, and any strict prefix that contains the same number of \texttt{.all} becomes \texttt{NO} permission. Any strict prefix that does not contain the same number of \texttt{.all} becomes write-only."} 
	
	In the other case, any subtree that has more indirections than the moved path is aliased. Hence those subtrees have their permission changed to \texttt{NO}, while all other nodes of the moved tree only lose their readability. Similarly, all ancestors have their permission set to \texttt{W}.\footnote{The exact wording as appearing in the SPARK RM says: ``After moving a path without an \texttt{‘Access} attribute, the path and any of its prefixes becomes write-only. Any shallow extension of the path that has the same number of \texttt{.all} stays read-write. Any deep extension of the path that has the same number of \texttt{.all} becomes write-only. Any extension of the path that has not the same number of \texttt{.all} becomes \texttt{NO} permission."}
	
	Finally, the assigned path is checked to have either \texttt{W} or \texttt{RW} permission. The assigned path and any of its extensions get \texttt{RW} permission. We then propagate this update to prefixes by increasing their permission with the greatest lower bound of the permissions of their descendants. More details on this propagation are given in Appendix~\ref{sec:PermRelease}, which defines the \textit{PermRelease} operator. 
	
	\subsection{Borrowing}
	
	Borrowing consists in a temporary transfer of ownership of a path from a caller to a callee. They happen only during subprogram call statements. A parameter is said to be borrowed when it is a name of mode \texttt{in-out} or \texttt{out},\footnote{Only names can be sent to procedures with mode \texttt{in-out} or \texttt{out}.} or of mode \texttt{in} and of type access-to-variable.\footnote{See Ada RM, section 3.10.(8) and 3.10.(10) \cite{adarm}.}
	
	The last condition allows sending pointers with mode \texttt{in}, and modifying their referenced value while preventing any modification to the pointer. Ada provides a modifier \texttt{constant} that forbids any modification to the referenced value, and can be used in order send pointers with mode \texttt{in} without granting the permission to write on those values.
	
	Intuitively, borrows can be seen as moving the borrowed path (actual parameter) with an \texttt{Access} attribute to the formal parameter, and then moving back the formal parameter to the actual. However, given that mode \texttt{in} parameters are not necessarily names, this complicates handling the borrows of such parameters, the details being explained in section~\ref{sec:Operational semantics}.
	
	The formal definition of a borrowed expression is similar to move, except that we do not require the top level expression to be deep for modes \texttt{in-out} or \texttt{out} (which coincides well with enclosing the top-level expression in a \texttt{Access} attribute). Hence, an expression in a procedure call is said to be borrowed if:
	\begin{itemize}
		\item It is an \texttt{in} actual parameter of a procedure and is of type access-to-variable.
		\item It is an \texttt{in-out} or \texttt{out} actual parameter to a procedure.
		\item It is the prefix of an \texttt{Access} attribute of a borrowed expression.
		\item It is a deep type direct subexpression of a borrowed expression. 
	\end{itemize}

	After gathering the borrowed paths, we check sequentially that they have \texttt{RW} permission (or at least \texttt{W} for \texttt{out} parameters), and then set any prefix or extension to \texttt{NO} for the whole duration of the call if it is of deep type, \texttt{R} if shallow (so that its value can be used in other arguments without being passed to the callee). At the end of the call, every borrowed path has its permission updated to \texttt{RW}, and as for assigns, we propagate this update to the prefixes of those paths.

	In the callee, any borrowed formal parameter of a procedure with mode \texttt{in} or \texttt{in-out} has \texttt{RW} permission, as well as any of its extensions. Any borrowed formal parameter of a procedure with mode \texttt{out} has \texttt{W} permission, as well as any of its extensions. We also require that at the end of the subprogram, every borrowed parameter should have permission \texttt{RW}, so that they could be moved back to the caller.

	\subsection{Observing}
	Observing is the mechanism that creates read-only aliases for parameters of mode \texttt{in} to subprograms. Formally any deep actual parameter of mode \texttt{in} that is not of type access-to-variable is observed, as well as any deep subexpression of an observed expression.

	After gathering the observed paths, we check sequentially that they have at least \texttt{R} permission and then set any prefix or extension to \texttt{R} for the whole duration of the call. At the end of the call, every observed path has its permission reverted to the one it had before the call.

	In the callee, any observed formal parameter of a procedure, as well as any of its extensions, have \texttt{R} permission.

	\subsection{Control structures}
	
	The previously given rules describe the evolution of permissions at the level of one block of statements. They should be completed with others rules that describe the checks done at different control structures available in SPARK, namely loops and conditions (gotos and exceptions are not available in SPARK). We only include conditions in {\muSPARK} and exclude loops.
	
	For conditions, we apply our rules on each block independently, yielding to different environments, and then merge those by taking, for each path, the greatest lower bound of its permission in the different environments.
	
	For loops, we require permissions to get only less restrictive at the end of the loop, when compared to the entry. This rule forbids moving the same variable to a different element of an array at each iteration.
	
	\section{Proofs of safety}
	
	The section that follows provides a proof of the previously given anti-aliasing rules for the {\muSPARK} language. For this purpose, we first provide a formalization of the {\muSPARK} language with a grammar and typing rules, then with an operational semantics, and finally a formalization of the previously given rules on which we provide a proof of non-aliasing.
	
	\subsection{Grammar and syntax}
	
	We remind that a {\muSPARK} program consists of a single procedure file containing declarations and body. A declaration can either be a procedure, a record type definition or an uninitialized variable declaration. The body consists of a sequence of instructions, that can be assignment statements of existing or freshly allocated values, procedure calls or if statements from which the condition has been abstracted (non-deterministic choice). 
	
	Only integer and user-defined record types are available, as well as accesses to existing types. Note that there are no loops, nor expression operators, which make the language non Turing-complete. 
	
	The full details of the grammar are given in Appendix~\ref{sec:Grammar}.
	
	\subsection{Operational semantics}
	\label{sec:Operational semantics}
	
	\subsubsection{Memory trees}
	The operational semantics is based on a memory model that keeps some relational information about the path leading to the designated memory area. To each path, we associate a memory cell with its value. This leads to the concept of \emph{memory trees} that are quite similar to permission trees with the exception that those trees are finite, given that the access nodes can point to null value. Hence, the \emph{memory environment} $\Upsilon$ is a mapping from variable declarations to memory trees. 
	
	We decide not to represent the implementation of aggregate structures inside the memory, supposed to be compiler-dependent. Thus we keep in our memory trees only a mapping of fields to children subtrees, as if one memory cell is enough for a structure to point to all its children. Indeed, without loss of generality, we can always chose in {\muSPARK} a word size big enough, so that each memory cell contains arbitrarily many pointers. 
	
	Each memory tree is defined in the following way:
	$$ \begin{array}{lcl}
	M & ::= & Integer(Cell, Value) \\
	& | & Record(Cell, Fields \rightarrow M) \\
	& | & Access(Cell, M) \\
	& | & Access(Cell, Null) \\
	& & \\
	\end{array} $$
	
	\subsubsection{Fresh(tau)}
	\label{sec:App-Fresh}
	Formally, memory cells are elements of an arbitrary infinite set. We use the word $\textit{fresh}$ to designate an oracle that gives a new element of the infinite set (seen as an allocation of a machine word). This construct can be easily generalized to allocate a whole tree representing a given type $\tau$, which we call $\textit{fresh}(\tau)$. It takes a type $\tau$ as an input and allocates the required memory areas for the type $\tau$ using the $\textit{fresh}$ oracle, and initializes them with default value.
	
	\begin{itemize}
		\item $\textit{fresh}(\texttt{integer})$ is equivalent to $Integer(\textit{fresh})$.
		\item $\textit{fresh}(\texttt{access}~ \tau)$ is equivalent to $Access(\textit{fresh}, Null)$.
		\item $\textit{fresh}(R)$ is equivalent to $Record(\textit{fresh}, \forall x:\tau \in R ~x\mapsto \textit{fresh}(\tau))$.
	\end{itemize} 
	
	\subsubsection{Assign}
	
	\begin{figure}[!h]
		\centering
		\footnotesize
		\begin{subfigure}[t]{0.3\textwidth}
			\raisebox{18mm}{\includegraphics[width=\textwidth]{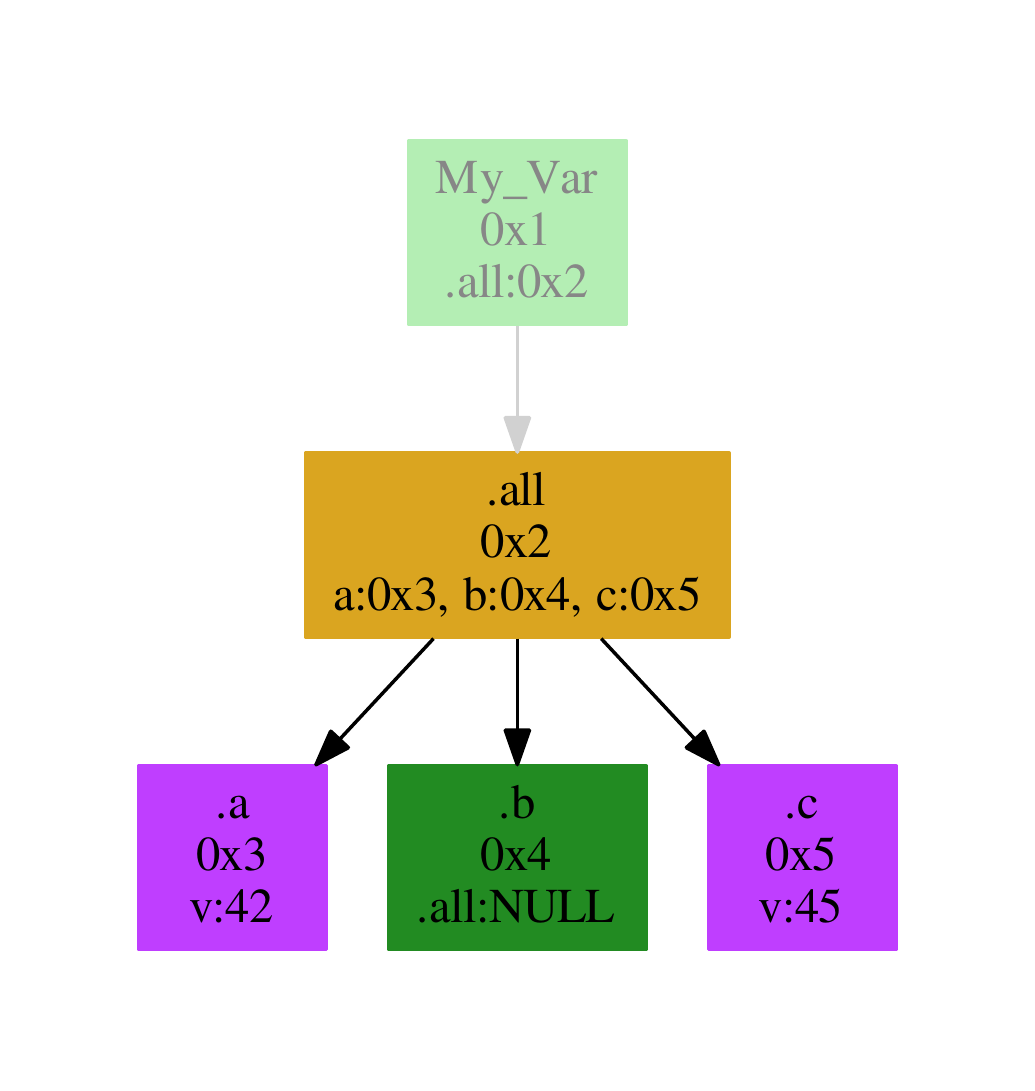}}
			\centering
			\caption{The tree of the left hand side \texttt{My\_Var}. The assigned part (\texttt{My\_Var.all}) is a structure containing three fields, \texttt{a} of type integer with value 42, \texttt{b} of type pointer to NULL, and \texttt{c} of type integer with value 45.}
			\label{fig:semlefthandside}
		\end{subfigure}
		\rulesep
		\begin{subfigure}[t]{0.3\textwidth}
			\includegraphics[width=\textwidth]{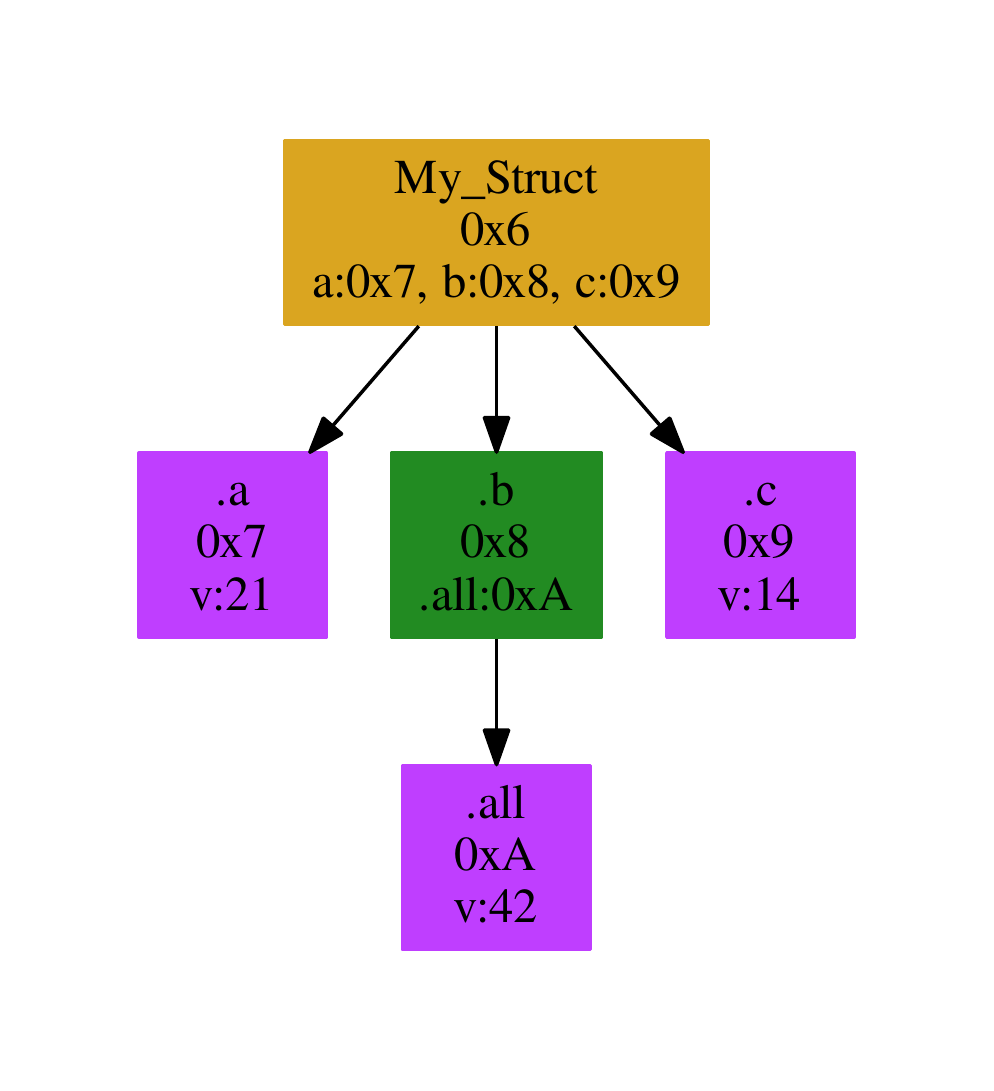}
			\centering
			\caption{The right-hand side \texttt{My\_Struct}. It is a structure containing three fields, \texttt{a} of type integer with value 21, \texttt{b} of type pointer to an integer with value 42, and \texttt{c} of type integer with value 14.}
			\label{fig:semrighthandside}
		\end{subfigure}
		\rulesep
		\begin{subfigure}[t]{0.3\textwidth}
			\includegraphics[width=\textwidth]{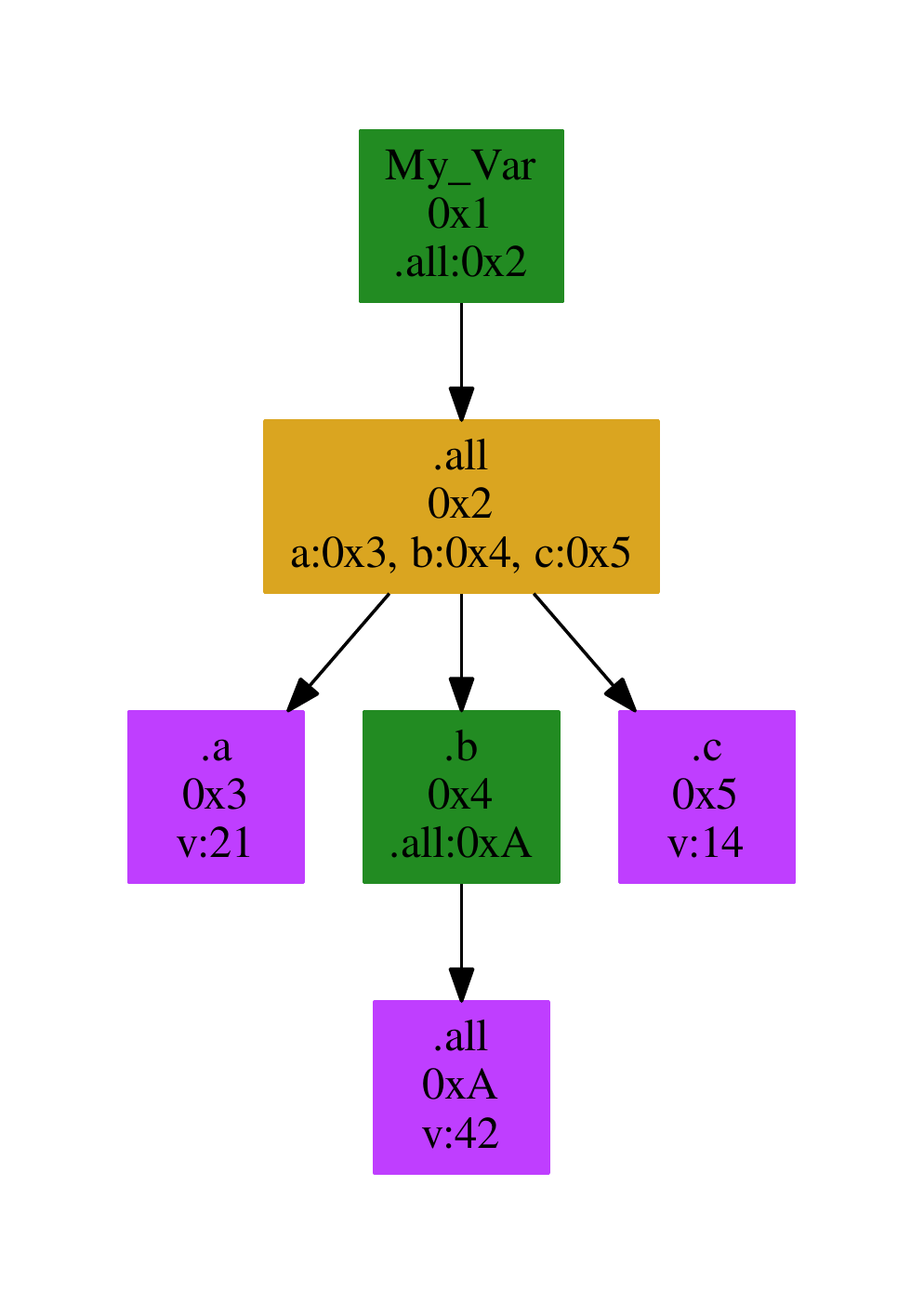}
			\centering
			\caption{The left-hand side after assignment. Note that the values of each node is updated without changing their addresses except for the indirected integer node (address \texttt{0xA}) that is added to the tree.}
			\label{fig:semlefthandsideafter}
		\end{subfigure}
		\caption{Example of the semantic rule \textsc{(E-assignName)} applied to the assignment \texttt{My\_Var.all := My\_Struct;} Pointers are in green, structures in gold, and values in purple.}
		\label{fig:semassign}
	\end{figure}

	The assignment of one subtree to a node is peculiar. Indeed, as in real Ada/SPARK, the left-hand side does not get its memory address changed, only its value. In our representation of the memory this would look like a (deep) replacement of every value in $Integer$ nodes, as well as a replacement of the subtree pointed by every $Access$ node.
	
	This leads to the  following definition of the recursive $\textit{Assign}_\Upsilon(Node, Node)$ operator (the two parameters are of same type):
	\begin{itemize}
		\item $\textit{Assign}(Integer(C, \_), Integer(\_, V)) = Integer(C, V)$
		\item $\textit{Assign}(Access(C, \_), Access(\_, P)) = Access(C, P)$
		\item $\textit{Assign}(Record(C, (F_i)_i), Record(\_, (F'_i)_i)) = Record(C, (\textit{Assign}(F_i, F'_i))_i)$
	\end{itemize}

	For assignments, we copy the value of the left-hand side to the right-hand side. In terms of semantics, this is done by copying the value of each field, except for pointers where we copy the whole subtree in the indirection, as shown in Figure~\ref{fig:semassign}.

	\subsubsection{Procedure calls}
	
	Another peculiarity happens at procedure calls, where we have to copy the actual parameter into a fresh tree representing the formal parameter (\textit{GetFromExpr}, see below), run the semantics on the procedure definition (yielding $\Upsilon''$), and then, depending on mode, copy back its content to the actual parameter. This semantics follows the compiled code, which pushes the parameters in the stack before call, jumps to the callee, and then pops them while moving them to the appropriate registers if those parameters are of mode \texttt{out} or \texttt{in-out}. This can be summarized by the rule \textsc{(E-call)}, whose definition can be found in section~\ref{sec:App-Call}.
	
	The formal definition of $\textit{GetFromExpr}_\Upsilon(e)$ is:
	\begin{itemize}
		\item $\textit{GetFromExpr}_\Upsilon(\texttt{null}) = Access(fresh, Null)$
		\item $\textit{GetFromExpr}_\Upsilon(literal) = Integer(fresh)$
		\item $\textit{GetFromExpr}_\Upsilon(n\texttt{'Access}) = Access(fresh, \Upsilon(n))$
		\item $\textit{GetFromExpr}_\Upsilon(n) = \Upsilon(n)$
	\end{itemize}
	
	We also use $ SetFromExpr_{\Upsilon}(e_i,x_i)$ for updating borrowed parameters with mode \texttt{in} (access-to-variable, note that mode \texttt{in} actual parameters are not solely names, hence the expression needs to be decomposed):
	\begin{itemize}
		\item $ SetFromExpr_{\Upsilon}(\texttt{null},x_i)$: no update
		\item $ SetFromExpr_{\Upsilon}(literal,x_i)$: impossible case
		\item $ SetFromExpr_{\Upsilon}(n,x_i) = Assign(n, \Upsilon(x_i))$
		\item $ SetFromExpr_{\Upsilon}(n\texttt{'Access},x_i) = Assign(n, \Upsilon(x_i.all))$ (valid since formal parameters of mode \texttt{in} cannot be set to \texttt{null} in callee).
	\end{itemize}

	\subsubsection{Semantics for statements}	
	
	Every rule is in the form $\Upsilon.c \xRightarrow{} \Upsilon'$, where $c$ designates any construct of the language. We annotate each semantics rule with $_s$ (statements) or $_d$ (declarations) for readability purposes.
	\label{sec:App-Call}
	\begin{multicols}{2}
		\infrule[E-assignNull]
		{} 
		{\Upsilon.~x \texttt{:= null} \xRightarrow{}_s \Upsilon[x.all\mapsto \texttt{Null}]}
		
		\infrule[E-assignName]
		{n~ \textnormal{name}} 
		{\Upsilon.~x \texttt{:=}~ n \xRightarrow{}_s\Upsilon[Assign(x, \Upsilon(n))]}
		
		\infrule[E-ifConditionTrue]
		{\Upsilon .~ i_1 \xRightarrow{}_s \Upsilon' } 
		{\Upsilon .~\texttt{if * then} ~ i_1 ~ \texttt{else} ~ i_2~ \texttt{end if} \xRightarrow{}_s \Upsilon'}
		
		\columnbreak
		
		\infrule[E-assignLiteral]
		{x:\tau\in\Gamma \andalso e~ \textnormal{literal}} 
		{\Upsilon.~x \texttt{:=}~ e \xRightarrow{}_s \Upsilon[x.value \mapsto e]}
		
		\infrule[E-assignAccess]
		{n~ \textnormal{name}} 
		{\Upsilon,\Phi.~x \texttt{:=}~ n\texttt{'Access} \xRightarrow{}_s \Upsilon[x.all\mapsto \Upsilon(n)]}

		\infrule[E-assignNew]
		{} 
		{\Upsilon.~x \texttt{:=}~ \texttt{new}~\tau \xRightarrow{}_s\Upsilon[x.all\mapsto \textit{fresh}(\tau)]}
		
		\infrule[E-block]
		{\forall j>0,~ \Upsilon_j .~ i_j \xRightarrow{}_s \Upsilon_{j+1}} 
		{\Upsilon_1. ~{\texttt{begin}~ i_1 ~...~ i_n~ \texttt{end}}  \xRightarrow{}_s \Upsilon_{n+1}}
		
		\infrule[E-ifConditionFalse]
		{\Upsilon .~ i_2 \xRightarrow{}_s \Upsilon' } 
		{\Upsilon .~\texttt{if * then} ~ i_1~ \texttt{else} ~ i_2~ \texttt{end if} \xRightarrow{}_s \Upsilon'}

		\infrule[E-call]
		{e_{1}..e_a~ \textnormal{with mode \texttt{in}}, \\ e_{a+1}...e_b~ \textnormal{with mode \texttt{in-out}}, \\~ e_{b+1}...e_n~ \textnormal{with mode \texttt{out}}, \\ \Upsilon'=\{\forall~ 0 < i \leq n, x_i \mapsto \textit{GetFromExpr}_{\Upsilon}(e_i) \} \\
			\Upsilon'.~\texttt{procedure} ~P(x_1,~...~,x_n) \xRightarrow{}_d \Upsilon'' \\
			\Upsilon'''=\Upsilon[\forall~ 0 < i \leq a, \textit{SetFromExpr}_{\Upsilon''}(e_i,x_i)\\
			\forall~ a < i \leq b, Assign(e_i, \Upsilon''(x_i)) \\
			\forall~ b < i \leq n, Assign(e_i, \Upsilon''(x_i))] 
		}
		{\Upsilon.~P(e_1, ~...~, e_n) \xRightarrow{}_s \Upsilon'''}
	\end{multicols}
	
	\subsubsection{Semantics for declarations}

	\infrule[E-uninitDecl]
	{x:\tau\in\Gamma} 
	{\Upsilon.~ x : \tau \xRightarrow{}_d \Upsilon[x\mapsto \textit{fresh}(\tau)]}

	\infrule[E-procedureDecl]
	{
		\Upsilon_1 = \Upsilon \andalso	\forall k, \Upsilon_k . ~d_k \xRightarrow{}_d \Upsilon_{k+1} \andalso
		\Upsilon_{m+1} . ~i \xRightarrow{}_d \Upsilon' }
	{\Upsilon .~\texttt{procedure} ~P(x_1, ... , x_n) ~\texttt{is} ~d_1, ..., d_m ~\texttt{begin}~ i~ \texttt{end} \xRightarrow{}_d \Upsilon'}

	\subsection{Permission rules}
	
	Similarly, the rules defined in section~\ref{sec:Rules} are mathematically formalized, with semantic-like rules that show the evolution of permission trees (defined in section~\ref{sec:Permission}) depending on lexical elements of the language and the mechanism used (borrowed, observed, moved). As for the semantics, we define some constructs to ``allocate''\footnote{Technically, the trees are not allocated, but only constructed by the \textit{Pfresh} operator.} (\textit{Pfresh}), normalize (\textit{PermRelease}), merge (\textit{Fusion}) permission trees. 
	
	As an example, the same assignment shown in Figure~\ref{fig:semassign} would give the following sequence using the rule \textsc{(P-assignDeepName)} (all of them can be found in Appendix~\ref{sec:Permission Rules for MuSPARK}). The whole derivation tree is presented hereafter, for better readability. 
	
	The rule \textsc{(P-assignDeepName)} applied to the assignment \texttt{My\_Var.all := My\_Struct} checks first that \texttt{My\_Var.all} is indeed a well defined deep variable, and that \texttt{My\_Struct} is a indeed a name. The permission corresponding to the node \texttt{My\_Struct} is checked to be Read-Write (using the notation $\Phi(\texttt{My\_Struct})$). Then, we apply the rule \textsc{(P-M-ident)} on the moved name \texttt{My\_Struct} that modifies its permission to write-only. The next part of the \textsc{(P-assignDeepName)} rule changes the permission of every extension of the moved name by putting its strict deep extensions with more indirections to \texttt{NO} permission (\texttt{My\_Struct.b.all}), strict deep extensions with same number of indirections of write-only (\texttt{My\_Struct.b}), and the strict shallow extensions to read-write (\texttt{My\_Struct.a}, \texttt{My\_Struct.c}).\footnote{The \textit{Readability} lemma shows that if a node has read permission, then all its children have also read permission. This gives as a corollary that read-write permissions follow the same rule outside procedure calls. Without this lemma, the changes of permission in the semantics should be interpreted as taking the least upper bound of the given permission and the actual permission. This method is the one implemented in Ada, being more generic.} 
	
	Only after those changes, we check for write permission to the assigned name (\texttt{My\_Var.all}). This order is very important so as to check the safety of assignments of variables to themselves and prevent creating cycles in data-structures (such as assigning \texttt{Tree.left.all := Tree}). Finally, we change the permissions of the assigned path to read-write, as well as every extension, and propagate this update to its prefixes with the \textit{PermRelease} operator that normalizes the permission tree.

	{\small
	$
	\inferrule* [Right=(P-assignDeepName)]
	{\texttt{My\_Var.all} : \texttt{S} ~\textnormal{deep} \in \Gamma \andalso \Phi(\texttt{My\_Struct}) = \texttt{RW} \andalso \inferrule* [ Right=(P-M-ident)] { } {\Phi.~\texttt{My\_Struct} \xrightarrow{Move}_n \Phi[\texttt{My\_Struct}\rightarrow \texttt{W}]} \\ \Phi'' = \Phi[\texttt{My\_Struct}\rightarrow \texttt{W}][\texttt{My\_Struct.b.all} \mapsto \texttt{NO} \andalso \texttt{My\_Struct.b} \mapsto \texttt{W} \\ \texttt{My\_Struct.a} \mapsto \texttt{RW} \andalso \texttt{My\_Struct.c} \mapsto \texttt{RW} ] \andalso \Phi''(\texttt{My\_Var.all}) = \texttt{W,RW} \andalso } 
	{\Phi.~\texttt{My\_Var.all} \texttt{:=}~ \texttt{My\_Struct} \xrightarrow{}_s \textit{PermRelease}(\Phi''[\texttt{My\_Var.all}, \\ \texttt{My\_Var.all.a}, \texttt{My\_Var.all.b}, \texttt{My\_Var.all.b.all}, \texttt{My\_Var.all.c} \mapsto \texttt{RW}])}
	$
}

	The Figure~\ref{fig:assignperm} shows the two permission trees associated with \texttt{My\_Var} and \texttt{My\_Struct} before and after the assignment \texttt{My\_Var.all := My\_Struct}. Given that the situation is symmetrical, the modification in permissions admits as an inverse the reverse assignment \texttt{My\_Struct := My\_Var.all}.
	
	\begin{figure}[!h]
		\centering
		\footnotesize
		\begin{subfigure}[t]{0.23\textwidth}
			\includegraphics[width=\textwidth]{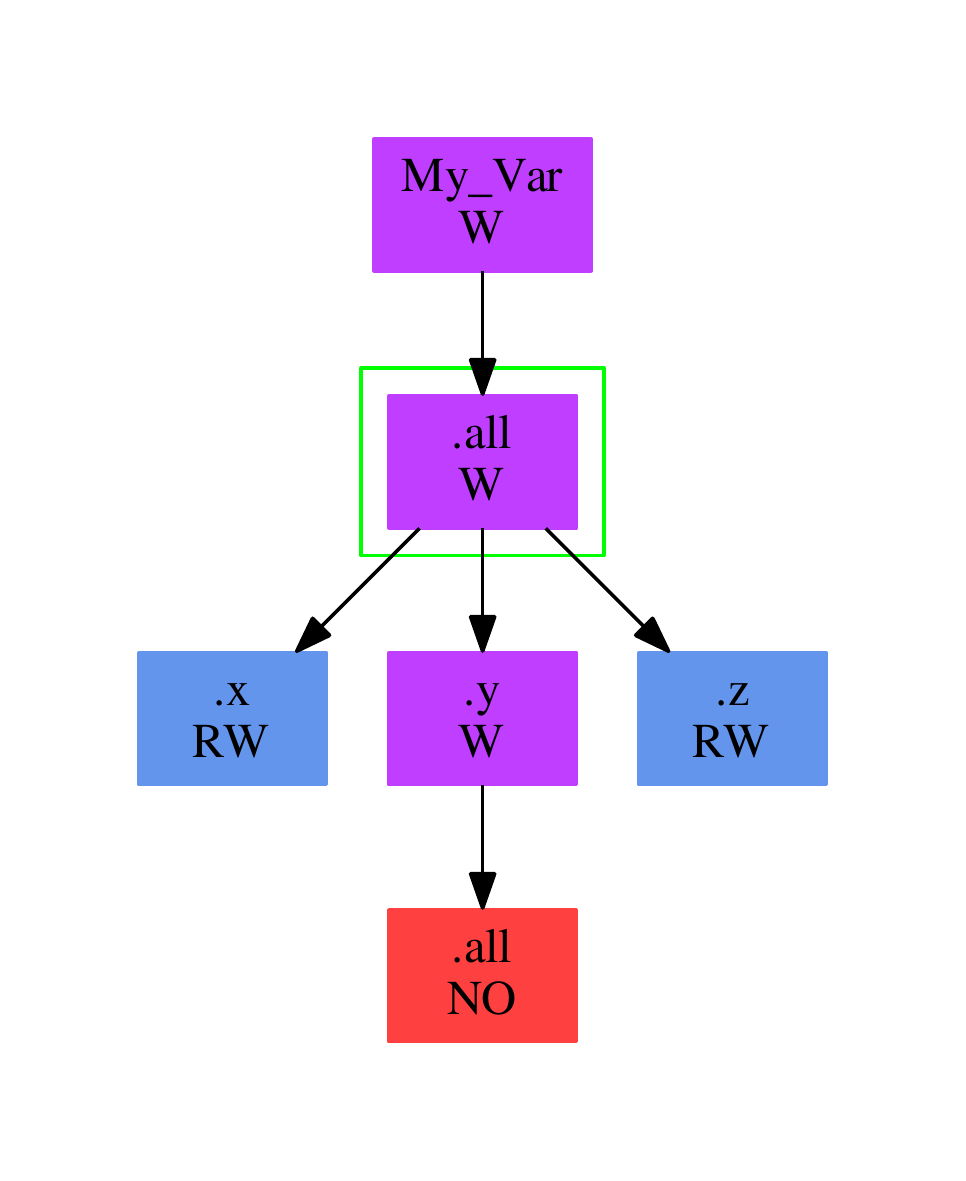}
			\centering
			\caption{The tree of the left hand side \texttt{My\_Var} before assignment. The assigned part (\texttt{My\_Var.all}) is shown in green.}
			\label{fig:assignleft}
		\end{subfigure}
		\rulesep
		\begin{subfigure}[t]{0.23\textwidth}
			\includegraphics[width=\textwidth]{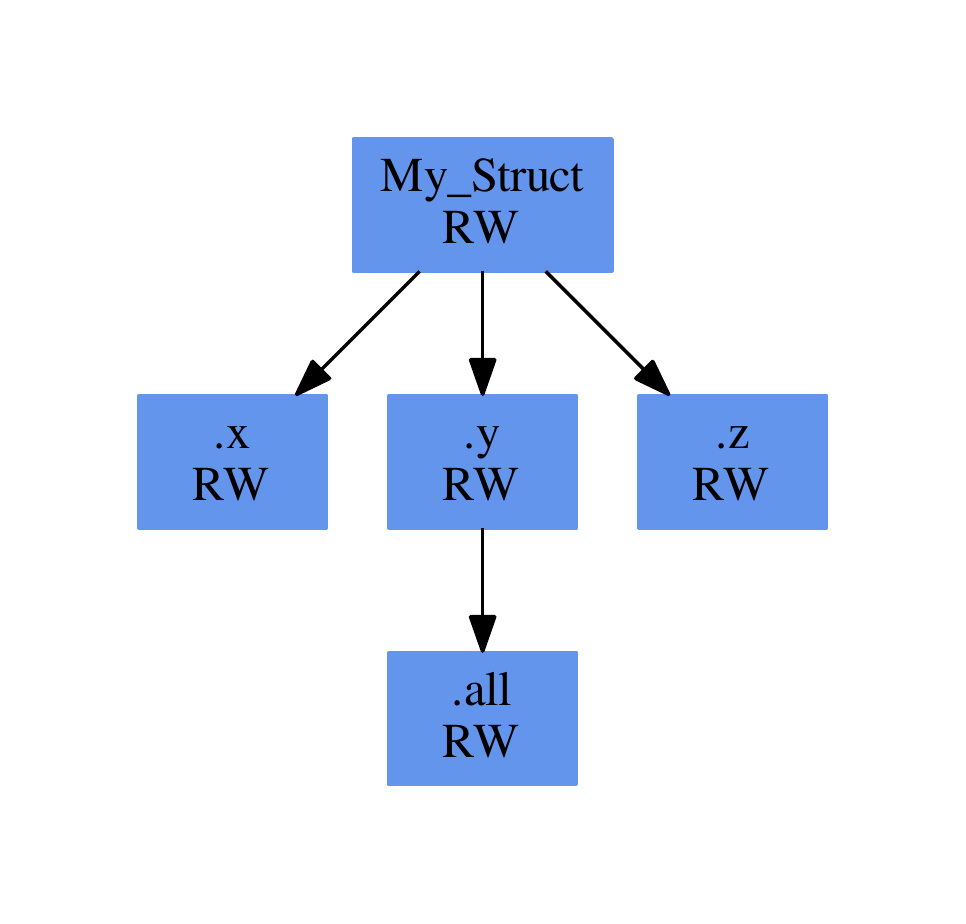}
			\centering
			\caption{The tree of the right hand side \texttt{My\_Struct} before assignment. }
			\label{fig:assignleftafter}
		\end{subfigure}
	\rulesep
\begin{subfigure}[t]{0.23\textwidth}
\includegraphics[width=\textwidth]{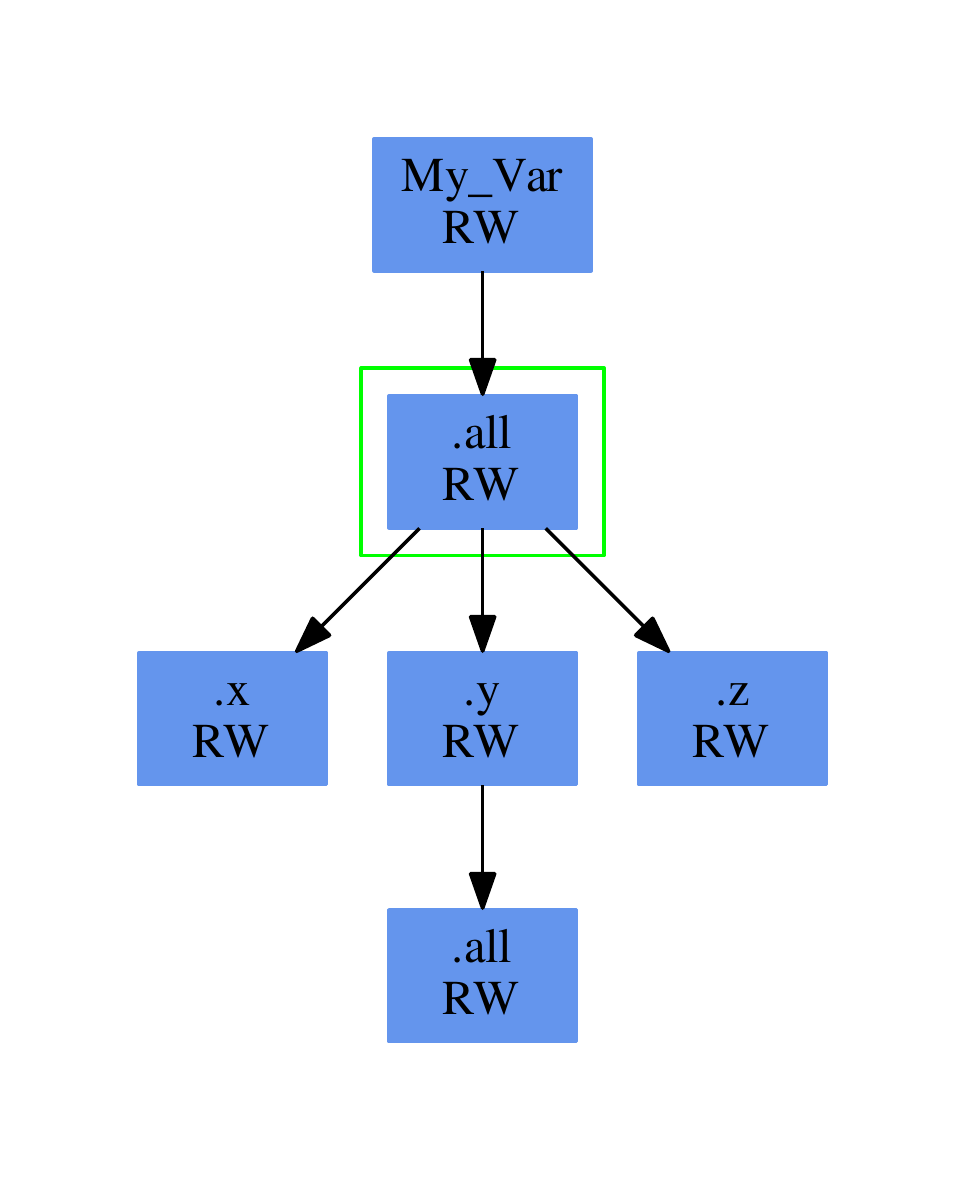}
\centering
\caption{The tree of the left hand side \texttt{My\_Var} after assignment. The \textit{PermRelease} operator changed the permission of \texttt{My\_Var} to \texttt{RW}.}
\label{fig:assignrightafter}
\end{subfigure}
\rulesep
\begin{subfigure}[t]{0.23\textwidth}
\includegraphics[width=\textwidth]{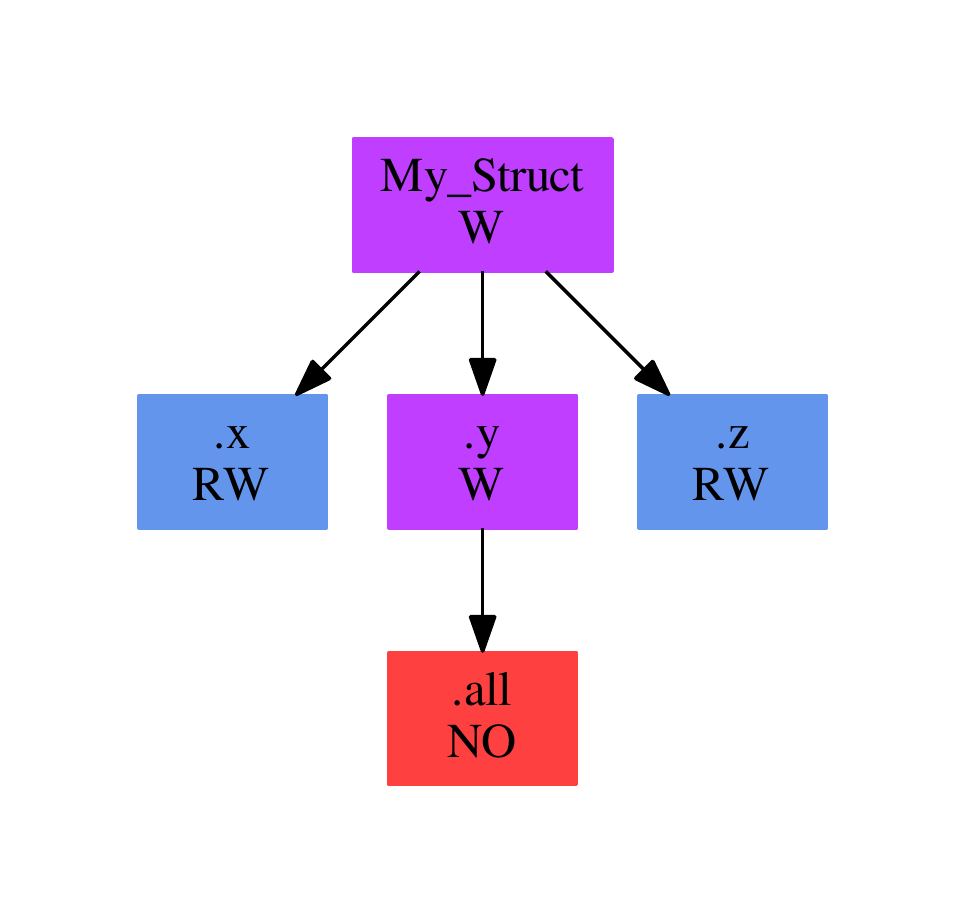}
\centering
\caption{The tree of the right hand side \texttt{My\_Struct} after assignment.}
\label{fig:assignright}
\end{subfigure}
		\caption{Example of the permission rule \textsc{(P-assignDeepName)} applied to the assignment \texttt{My\_Var.all := My\_Struct;}. }
		\label{fig:assignperm}
	\end{figure}

	As for the operational semantics, the rule for procedure calls is also special. Depending on each argument mode, we apply a different rule (either \textit{Observe}, \textit{BorrowIntOut} or \textit{BorrowOut}; the last differs from the previous by accepting write-only permission to the borrowed actual parameter). Then for every \texttt{out} or \texttt{in-out} parameter, they are set to read-write as well as all their extensions, as for assignment. Note that all the changes in permissions during call last only for the duration of the call.
	
	\infrule[P-call]
	{\texttt{procedure} ~P(x_1:\tau_1,~...~,x_n:\tau_n) ~\textnormal{with}~\\ x_1...x_a~ \textnormal{observed},\\ ~ x_{a+1}...x_b~ \textnormal{borrowed with mode \texttt{in}}, \\ ~ x_{b+1}...x_c~ \textnormal{borrowed with mode \texttt{in-out}}, \\~ x_{c+1}...x_n~ \textnormal{borrowed with mode \texttt{out}}, \\ 
		\Phi_1 = \Phi \\ 
		\forall~ 0 < i \leq a, \Phi_i.~e_i \xrightarrow{Observe}_e \Phi_{i+1} \\ 
		\forall~ a < i \leq b, \Phi_i.~e_i \xrightarrow{BorrowInOut}_e \Phi_{i+1} \\
		\forall~ b < i \leq c, \Phi_i.~e_i \xrightarrow{BorrowInOut}_e \Phi_{i+1} \\   
		\forall~ c < i \leq n, \Phi_i.~e_i \xrightarrow{BorrowOut}_e \Phi_{i+1} \\
		\Phi'=\textit{PermRelease}(\Phi[\forall~b < i \leq n: e_i \mapsto  \textit{Pfresh}(\tau_i,\texttt{RW})])
	}
	{\Phi.~P(e_1, ~...~, e_n) \xrightarrow{}_s \Phi'}
	
	\subsection{Correctness of permission rules with respect to the operational semantics}
	
	The anti-aliasing requirement enforces CREW. This means that if an aliased memory node can be written (by assigning its associated path, or any prefix with the same number of indirections), then any other aliased node can be neither written nor read using another path. 
	
	The way we built our memory representation adds to each memory cell the path used to access it, thanks to the tree representation. This allows tracking aliasing very precisely, which is by definition the fact of accessing the same memory cell using two different paths. This gives a straightforward formalization of the main theorem.
	
	\setcounter{theorem}{0}
	\begin{theorem}\textnormal{(No-aliasing)\textbf{.}}
		For every set of nodes $S$ in memory environment $\Upsilon$ such that their memory cells have the same address, consider their associated paths. If there is one path that can be written using the permission environment $\Phi$, then all other paths in $S$ can neither be written nor read.
	\end{theorem}
	
	\subsubsection{Lemmas}
	
	Some lemmas are required, so as to prove that an associated path to a memory node always exists. The solution for this is to compare memory and permission trees, and find out that memory trees are a subset of permission trees. This is the \emph{coherence} lemma, whose formalization is given hereafter.
	
	There are other lemmas, used in the proofs, such as the fact that each permission rule leaves the environment normalized (\emph{Normalization} lemma), that if a node has permission read, then all its children nodes also have read permission (\emph{Readability} lemma), and that the memory cannot loop on itself (\emph{No-cycle} lemma).

	\begin{lemma}\textnormal{(Normalization)\textbf{.}}
		Permission environment is normalized. 
	\end{lemma}
	\begin{proof}
		Straightforward. Every rule of the semantics calls the normalization operator \textit{PermRelease}. The rules that do not call it are \textsc{(P-block)} (trivial case), and declaration rules (they create only new trees with one node).
	\end{proof}
	
	\begin{lemma}\textnormal{(Coherence)\textbf{.}}
		For every node $n$ in the memory environment, it is possible to find an associated node $n'$ in the permission environment, such that $n$ and $n'$ designate the same path.
	\end{lemma}
	\begin{proof}
		Every permission rule manipulates permission trees that have the same constructs as memory trees, except for access types that cannot be null. Hence every memory tree can be obtained from a permission tree by cutting nodes at access nodes when the pointer has \texttt{null} value.
	\end{proof}

	\begin{lemma}\textnormal{(Readability)\textbf{.}} 
		If a node has permission \texttt{R} (resp. \texttt{RW}), then all its children have permission \texttt{R} (resp. \texttt{RW}) at each step of the semantics.
	\end{lemma}
	
	\begin{proof}
		By induction on the semantics: 
		\begin{enumerate}
			\item \textsc{(E-assignNull)}, \textsc{(E-assignLiteral)}: $\textit{Pfresh}(\tau, \texttt{RW})$ guarantees it for the assigned path.
			\item \textsc{(E-assignName)}, \textsc{(E-assignAccess)}, \textsc{(E-assignNew)}: same thing for the assigned path. For the moved path, every prefix gets its permission set to either \texttt{W} or \texttt{NO} by the \textsc{(P-M-ident)}, \textsc{(P-M-field)}, \textsc{(P-M-deref)} or \textsc{(P-SM-ident)}, \textsc{(P-SM-field)}, \textsc{(P-SM-deref)} rules. The extensions are handled by \textsc{(P-assignDeepName)}, \textsc{(P-assignShallowName)}, \textsc{(P-assignAccess)}, and \textsc{(P-assignNew)}.
			\item \textsc{(E-call)}: we create such nodes by observing their paths. The rules \textsc{(P-O-nullValue)}, \textsc{(P-O-litteral)}, \textsc{(P-O-takeAccess)}, \textsc{(P-O-name)} guarantee that if any node is set to permission \texttt{R}, then any extension of it is also set to \texttt{R}. After the callee returned, the proof is identical to assignments.
			\item \textsc{(E-block)}, \textsc{(E-ifConditionTrue)}, \textsc{(E-ifConditionFalse)}, \textsc{(E-uninitDecl)}, \textsc{(E-procedureDecl)}: trivial
		\end{enumerate}
	\end{proof}
	
	\begin{lemma}\textnormal{(No-cycle)\textbf{.}} 
		A node cannot have the same memory cell as any of its descendants (hence memory trees are finite).
	\end{lemma}
	
	\begin{proof}
		By induction on the semantics: 
		\begin{enumerate}
			\item \textsc{(E-assignNull)}, \textsc{(E-assignLiteral)}, \textsc{(E-assignNew)}: given that those rules only cut memory trees, there are no new indirections created.
			\item \textsc{(E-assignName)}: such a cycle could only be created when applying $Assign$ on a node $Access(C1, \_)$ which is a descendant of the node being moved $Access(\_, P)$. In such a case, this would contradict hypothesis of the rule \textsc{(P-assignDeepName)}, given that we set  descendants of the moved node (ie strict extensions) to \texttt{NO}, before checking \texttt{RW} for the node to assign.
			\item \textsc{(E-assignAccess)}: such a cycle could only be created if the assigned $Access$ node is a descendant of the subtree we are taking address of. In such a case, this would contradict hypothesis of the rule \textsc{(P-assignAccess)}, given that we set  descendants of the moved node (ie strict extensions) to \texttt{NO}, before checking \texttt{RW} for the node to assign.
			\item \textsc{(E-call)}: borrows and observes do not modify addresses before transferring to callee. After returning, we assign every parameter. Hence identical as \textsc{(E-assign)} on each \texttt{in-out} or \texttt{out} parameter. For every \texttt{in} parameter, \textit{SetFromExpr} is equivalent to either assigning $n.all$ or $n$.
			\item \textsc{(E-block)}, \textsc{(E-procedureDecl)}: trivial by applying the induction hypothesis successively to each statement of the block.
			\item \textsc{(E-ifConditionTrue)}, \textsc{(E-ifConditionFalse)}: the $Fusion$ operator does not change memory places and pointers, only permissions. 
			\item \textsc{(E-uninitDecl)}: every node is assigned to a new memory area (and pointers to null). Hence it is impossible to point to an existing memory area.
		\end{enumerate} 
		
	\end{proof}

	\subsubsection{No-aliasing proof}

\setcounter{theorem}{0}
\begin{theorem}\textnormal{(No-aliasing)\textbf{.}}
	For every set of nodes $S$ in memory environment $\Upsilon$ such that their memory cells have the same address, consider their associated paths. If there is one path that can be written using the permission environment $\Phi$, then all other paths in $S$ can neither be written nor read.
\end{theorem}

	The main theorem is proved by induction on the operational semantic. This leads us to consider each semantic rule one by one and show that the invariant holds. As an example, let us consider the proof for \textsc{(E-assignName)} and \textsc{(E-call)}. The full proof is given in Appendix~\ref{sec:Proof}.
	
	\paragraph{\textsc{(E-assignName)}} Let us take any set $S$ in the environments $\Upsilon',\Phi'$ after executing one step of the semantics. Nodes can only get aliased between the first \textit{Access} node descendant of assigned and moved node (because values are recursively copied, up to the first \textit{Access} node encountered in which the pointer is copied, which creates an alias). 
	
	We can consider only the case when the moved or assigned path $y$ is the one that can be written. Indeed, by \textit{Readability}, $y$ cannot be \texttt{R} in $\Phi$. Hence it can only be \texttt{NO}, which leads to a contradiction. If $y$ cannot be written nor read, then its permission is \texttt{NO} before executing the semantics. By normalization of $\Phi$, we can say that the glb of its children is also \texttt{NO} (and recursively). Which contradicts \texttt{RW} permission for $x$.
	
	Executing one step of the semantics creates an alias in $x$ (the assigned path) for every node of any subtree rooted under an $Access$ node that is a direct descendant of $n$. However, the rule \textsc{(P-assignDeepName)} says that when encountering an $Access$ node, all the further descendants are set to \texttt{NO}, hence their paths cannot be written. The rules \textsc{(P-M-ident)}, \textsc{(P-M-field)}, \textsc{(P-M-deref)} guarantee also that any extension of $n$ cannot be read anymore. This solves the case for elements of $S$ that are descendant of $x$ with different number of \texttt{.all}. 
	
	For elements of $S$ that are descendant of $x$ with same number of \texttt{.all}, no alias is created during assignment, hence induction hypothesis can be applied directly. 
	
	Hence, the assigned path gets \texttt{RW} permission (as well as any extensions), but any aliased subtree of this node cannot be written nor read in the moved tree, hence only one aliased path be written, and all others can neither be written nor read.
	
	\paragraph{\textsc{(E-call)}}
	
	The case for procedure calls is particular: indeed we apply our inductive hypothesis after modifying environments, hence we have to guarantee that the invariants not only hold at the end of the rule, but also at the moment we transfer the control flow to the callee (ie for $\Upsilon'$ in \textsc{(P-call)} and $\Phi'$ in \textsc{(P-procedureDecl)}).
	
	Let us consider such a set $S$ in $\Upsilon'$ (at the moment transferring to the callee). Suppose that one element of $S$ can be written. Then the rule \textsc{(P-procedureDecl)} guarantee that this element is borrowed (it is a borrow of a subtree $T$ of $\Upsilon$). Hence the rules \textsc{(P-B-name)}, \textsc{(P-B-name-Out)}, \textsc{(P-B-takeAccess)}, \textsc{(P-B-nullValue)}, guarantee that the borrowed subtree has permission set to \texttt{NO} if deep, and \texttt{R} if shallow. Note that the borrowed subtree cannot have been observed, given that we borrow after observing, hence we require \texttt{RW} permission on an argument that has been set to \texttt{R} by observation. Thus, it is impossible for $T$ to be observed, hence creating another element of $S$ that could be readable. For the same reason, the hypothesis of the rules \textsc{(P-B-entryPointInOut)} and  \textsc{(P-B-entryPointOut)} guarantee that the argument cannot be borrowed twice. Hence if one path can be written, all others cannot be neither written nor read.
	
	\section{Implementation and results}
	
	\subsection{Laziness of permission trees}
	
	The implementation of the permission rules is done in Ada. They are implemented as a separate module of the GNAT Pro compiler (more precisely in the front-end) and the analysis procedure is invoked from SPARK analyzer when called with the special flag \texttt{-gnatdF}. 
	
	The implementation is 6200 lines long. It involves dynamically allocated tree data structures that are used to implement the permission trees. However, given that those trees may be infinite, we decided to proceed with a lazy implementation of permission trees, with a special dethunking method. Indeed, the maximum depth at which those permission trees are used is the maximum number of lexemes of a path, which is always finite, and in practice less than 7. Moreover, the AST does not allow easy iteration on extensions of paths, hence those trees have to be built on the fly, leaving undefined many parameters.
	
	The definition of permission trees is hence modified to accept arrays, as well as fields to records that may not be referenced by the original definition, but may be added by object oriented programming (class-wide or incomplete types). Note that the $Thunk$ node may represent both a leaf and an unevaluated internal node.
	
	$$ \begin{array}{lcl}
	P & ::= & Thunk(Permission, Is\_Node\_Deep, Children\_Permission) \\
	& | & Record(Permission, Is\_Node\_Deep, Fields \rightarrow P, P) \\
	& | & Access(Permission, Is\_Node\_Deep, P) \\
	& | & Array(Permission, Is\_Node\_Deep, P) \\
	& & \\
	\end{array} $$
	
	This creates some approximations in our implementation, specially when setting the permissions to every extension in our permission rules, given that it is not possible to iterate over extensions. The exact implementation dethunks the tree depending on the type of the node, except for class-wide or incomplete types, that are replaced by an over-approximation. Similarly, when the permission changes for the whole subtree (such as assigning to a node), then the subtree is deallocated and replaced by a $Thunk$ node.
	
	\subsection{Complete SPARK}
	\label{sec:completeSPARK}
	
	\begin{wrapfigure}{r}{0.35\textwidth}
		\centering
		\vspace{-1cm}
		\includegraphics[width=.35\textwidth]{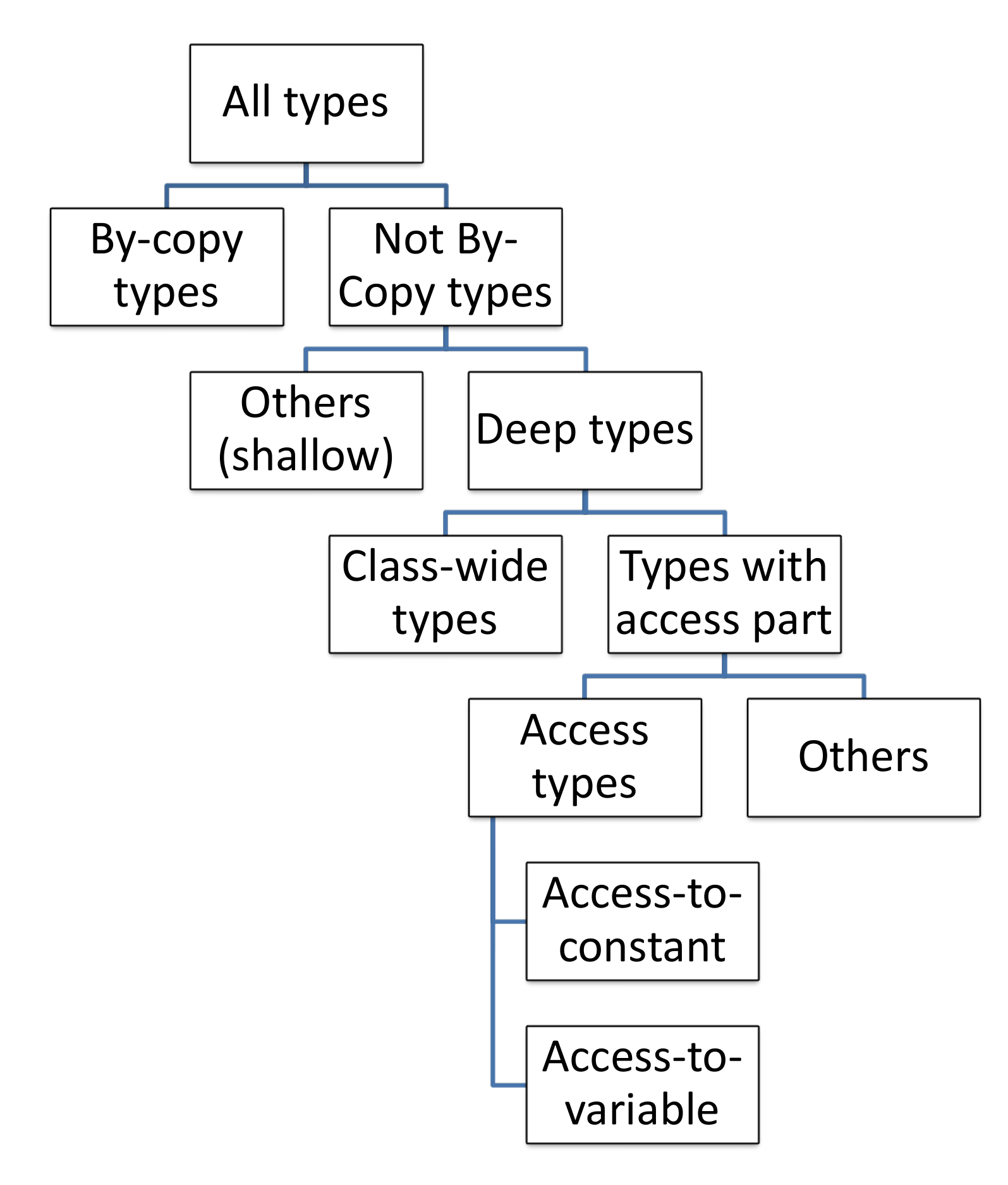}
		\caption{The different deep types in full Ada. Access types can be added the modifier \texttt{constant} that will prevent any modification to the pointed value. Any other access type is called access-to-variable.}
		\label{fig:deeptypes}
	\end{wrapfigure}

	The rules presented in this report only address a subset of SPARK. However, during the internship, the rules as well as the implementation target the complete SPARK added with access types. Complete SPARK differs from {\muSPARK} on several points.
	
	As evoked before, SPARK has arrays. We apply our permission rule to all elements, without taking into account the exact index of that element. That means when assigning to an element of an array, there are no effects in terms of permissions, given that it is not known which element has been assigned. In section~\ref{sec:Exemples}, we present a method to have an iterator over an array of pointers.
	
	Besides procedures, SPARK has functions, that can return values and whose calls are expressions (instead of statements for procedures). Note that functions in SPARK can only have parameters of mode \texttt{in} and cannot have side effects. To handle them safely, we add the rule that every formal parameter has read-only permission (observing). Moreover, the expression returned is considered to be moved. This constructs allows constructors (destructors being procedures), but not accessors.\footnote{However Ada has a mechanism called \emph{renaming} that allows renaming a path to another, acting like an accessor to one specific path.}
	
	Types for which passing to procedure as parameter can cause aliasing problems are not the same as the ones that can cause aliasing problems when being assigned. Actually deep types only designate the latter, the former being already defined in Ada as \emph{not-by copy} types (see Figure~\ref{fig:deeptypes}). Indeed, parameters to procedure are passed by-copy when they are shallow and can fit in a standard machine register. Other parameters \textbf{may} be passed by reference if shallow (compiler dependent) or \textbf{are} passed by reference if deep. Hence, the rules for complete SPARK replace deep by \emph{not by-copy} in observes and borrows, while deep is kept for moves.
	
	Many features from OOP have been ignored in this report, such as class-wide types. Some approximations are done to keep the safety of the analysis. We consider all class-wide types to be deep, by just ignoring their content, and considering them as always aliasing types.
	
	Loops in Ada can be either finite or infinite. Exit statements allow exiting any enclosing loop. Thus, we require that the permission of each path at the exit point of a loop (exit statement or end of finite loop), has to be less restrictive than at entry. This is enforced by an hashtable that associates each loop id to two permission environments, the one at entry and an accumulator that merges every environment at each exit point of the loop. 
	
	Global variables are considered as implicit parameters to the most enclosing procedure, all the more since SPARK requires specifying the mode of each used global in a procedure. Thus, they obey to the same rules as formal parameters and actual parameters. When calling a procedure, the caller has to ensure that every global variable used in the procedure has adequate permissions with respect to specified mode. 
	
	Packages in Ada can have elaboration code that is executed when being loaded from outside. Any initialized global declaration in Ada is implicitly rewritten as initialization code. This feature is very useful for interacting with hardware that needs to be initialized before any call to the library. In our analysis, we treat this code as a procedure that has as \texttt{out} parameter every stateful global variable of the package.

	\subsection{Test suites}
	
	The analysis has first been tested on two test suites. The first, called \emph{fixedbugs}, contains 17041 tests for all bugs fixed and features added in GNAT Pro compiler. The second is called \emph{acats-4} (Ada Conformity Assessment Test Suite),\footnote{They can be downloaded from http://www.ada-auth.org/acats-files/4.1/ACATS41.ZIP} and has 3905 standardized tests that every Ada compiler must pass. Those test suites must be passed in order to show that the new features implemented in the front-end do not break the existing compiler architecture. However, given that they do not contain any SPARK code, they do not allow assessing the efficiency of our rules on existing code base.
	
	The most interesting test suite is \emph{spark2014}, specific to SPARK. This suite has 2087 tests with valid SPARK code, and our analysis has only 30 regressions, almost all of them being caused by class-wide global variables that are manipulated with finer graining than the analysis is able to handle.
	
	Finally, a small test suite written by the author has been used to check that the different constructs act accordingly to the anti-aliasing rules. These tests are inspired by the examples given in Rust borrow-checker README file.\footnote{https://github.com/rust-lang/rust/blob/master/src/librustc\_borrowck/borrowck/README.md} 
	
	\subsection{Some use cases}
	\label{sec:Exemples}
	
	\subsubsection{Swap}
	\label{sec:Exswap}
	The first example is the swap procedure, whose naive implementation gets accepted by our rules. Note that there is no way of implementing a swap function in Rust. In Ada, we take advantage of \texttt{in-out} mode, that guarantees that after the procedure call, the \texttt{in-out} actual parameter to be assigned is exactly at the same address than the one being sent to the callee.
	\begin{lstlisting}[style=spark]
procedure Swap (X, Y:in out T) is
	Temp :T :=Y;  -- Move Y. X:RW, Y:W, Temp:RW
begin
	Y:=X;         -- Move X. X:W, Y:RW, Temp:RW
	X:=Temp;      -- Move Temp. X:RW, Y:RW, Temp:W
end Swap;       -- Both borrowed arguments X and Y have RW permission.
\end{lstlisting}
	\subsubsection{Iterator}
	The second example shows how to have a mutable iterator over an array of pointers using the previously defined swap procedure. The procedure guarantees \texttt{RW} permission for both its arguments. Note that it is also possible to use renaming declarations for this case, but only the swap method could iterate pairwise over an array (bubble sorting, ...).
\begin{lstlisting}[style=spark]
Iterator:=Null;  -- Iterator:RW
for i in a..b loop
	Swap(Iterator, My_Array(i));  -- Iterator:RW and My_Array(...):RW = Null
	DoStuff(Iterator);            -- Iterator:RW and My_Array(...):RW = Null
	Swap(Iterator, My_Array(i));  -- Iterator:RW = Null and My_Array(...):RW
end loop;  --  My_Array:RW
\end{lstlisting}
	\subsubsection{Dynamic data structures}
	\label{sec:dds}
	The last example shows some pieces of a code that manipulates trees with their child-sibling representation\cite{Fredman1986}. The code has a procedure \texttt{Free\_Node} that deallocates recursively a whole tree, and some statements that allocate a tree shown by Figure~\ref{fig:childsibling}. Then some procedure is called with two nodes passed as borrowed parameters with mode \texttt{in-out}, before the whole tree is freed.
	
	\pagebreak
	\begin{wrapfigure}{r}{0.33\textwidth}
		\centering
		\vspace{-1cm}
		\includegraphics[width=.35\textwidth]{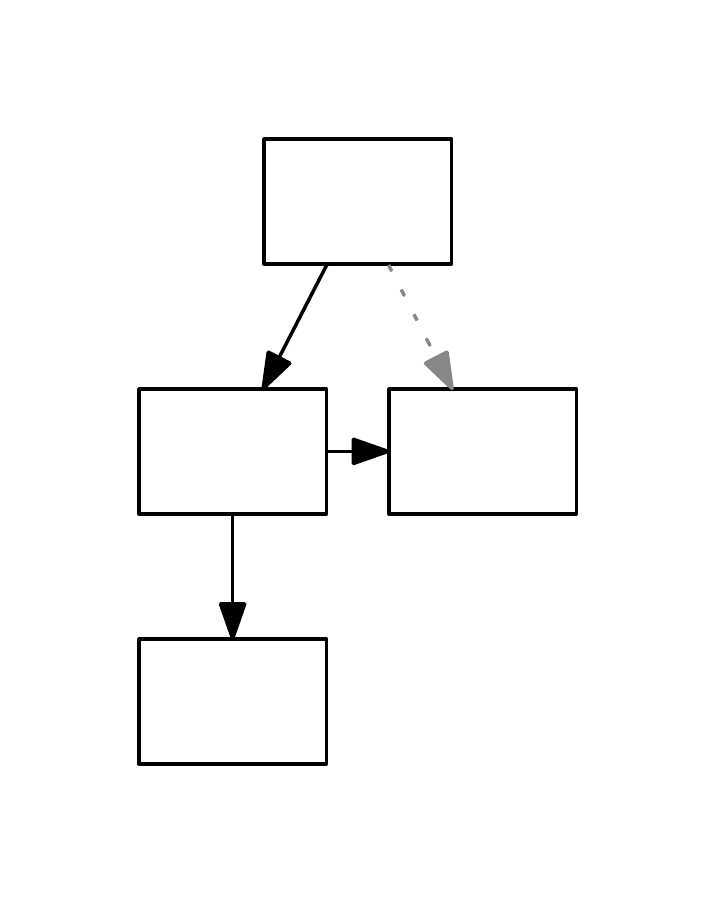}
		\caption{The tree in example \ref{sec:dds}.}
		\label{fig:childsibling}
		\vspace{-1cm}
	\end{wrapfigure}
	
\begin{lstlisting}[style=spark]
procedure Free_Node (N : in out Node) is
  if N.Child /= null
    Free_Node (N.Child);
  end if;
  if N.Sibling /= null
    Free_Node (N.Sibling);
  end if;
  Free (N); -- N set to null by Free;
end Free_Node;

Node := new Node;               -- allocate node
Node.Child := new Node;         -- allocate child
Node.Child.Child := new Node;   -- allocate child
Node.Child.Sibling := new Node; -- allocate sibling

Some_Procedure (Node.Child, Node.Sibling);

Free_Node (Node);

\end{lstlisting}
\vspace{-0.5cm}
	\subsection{Comparison with Rust}
	
	The following section compares Rust and SPARK on some constructs that seem relevant to the author. We study what features the anti-aliasing allows and what constructs can be done using these rules. The main difference comes from the fact that SPARK uses provers to handle a lot of safety features, that are handled directly by Rust's borrow-checker.
	
	\begin{wrapfigure}{r}{0.45\textwidth}
		\centering
		\vspace{-0.5cm}
		\begin{tabular}{|l|c|c|}
			\hline
			\textbf{Features} & \textbf{SPARK} & \textbf{Rust} \\
			\hline No-aliasing & \checkmark & \checkmark \\
			\hline Lifetime check & \tiny Specific checks & \checkmark \\
			\hline Automatic reclamation & \tiny Pools & \checkmark \\
			\hline Initialization checking & \tiny Flow & \checkmark \\
			\hline Nullity checking & \tiny By proof & \checkmark \\
			\hline Simple semantics & \checkmark & \cite{rustbelt} \\
			\hline No user annotation & \checkmark & \\
			\hline 
		\end{tabular}
		\caption{Comparison of the features available from the anti-aliasing analysis of SPARK and Rust.}
		\label{fig:features}
	\end{wrapfigure}

	First comparison is on the features done by the analysis of both languages (as shown in Figure~\ref{fig:features}). Even if both of them prevent non-benign aliasing in the source code, comparison shows that our rules do not handle as much features as Rust's borrow-checker, but allows more constructs than Rust (such as swap present in \ref{sec:Exswap}). 
	
	Indeed, Rust's borrow-checker checks lifetimes (so that a local deep variable cannot be assigned to a global one), whereas this check is implemented as a separate analysis of GNAT compiler, with some more restrictive rules than Rust. 
	
	Automatic reclamation is handled by Rust's borrow-checker (using a static garbage collector), whereas Ada has a separate feature, called \emph{pools}, acting like virtual stacks with a scope. Every variable allocated in this pool gets freed when the pool is destroyed (and lifetime checks prevent aliasing between two pools of different scopes). 
	
	Initialization checks are similarly done by Rust's borrow-checker in a safe way, whereas our rules only catch some uninitialized variable usage, with SPARK flow analysis doing the safe and very fine-grained analysis.
	
	Nullity-checking is also a built-in feature of Rust (there a no safe null pointers in Rust), whereas we allow them, considering dereferencing a null pointer should be considered as a runtime error (like division by 0), and should be checked by SPARK static analysis tool-suite at silver level (AoRTE proof) \cite{thalesada}.\footnote{The joint Thales-AdaCore guideline defines fives levels of software assurance for SPARK software. The first is stone, that consists of valid SPARK code. The second is bronze, SPARK code that passed the initialization check and has a correctly specified data-flow. The third is silver, that adds the proof of AoRTE. The fourth is gold, in which key integrity properties are proven, such as pre and post conditions, type predicates, loop invariants. The fifth requires full functional proof of requirements, with a program fully specified. More details can be found on the following blog post http://www.spark-2014.org/entries/detail/verifythis-challenge-in-spark }
	
	Finally, we should note that our analysis does not require any additional user annotation to the source code (whereas Rust needs sometimes explicit lifetimes to be specified), and the rules are unambiguously defined whereas for Rust no official document specifies the borrow-checker (the README provided is quite incomplete).

	\begin{wrapfigure}{r}{0.45\textwidth}
		\centering
		\vspace{-0.5cm}
		\begin{tabular}{|l|c|c|}
			\hline
			\textbf{Features} & \textbf{SPARK} & \textbf{Rust} \\
			\hline Move semantics & \checkmark & \checkmark \\
			\hline Read only grant & \tiny Only to callee & \checkmark \\
			\hline Move to callee & \checkmark & \checkmark \\
			\hline Move back to caller & \checkmark & \tiny Only 1 path\\
			\hline Primitives (swap) & \checkmark & \\
			\hline Fine graining (R/RW) & \checkmark &  \\
			\hline
		\end{tabular}
		\caption{Comparison of the constructs accepted by the anti-aliasing analysis of SPARK and Rust.}
		\label{fig:constructs}
	\end{wrapfigure}

	Second comparison is on what could be defined as expressiveness of the analysis, that means which constructs can be implemented on this analysis. Although we did not formalize this concept (and did not prove for Rust that some constructs could not be implemented), we consider that a little review of some constructs can be useful.
	
	We previously showed that cycling constructs could not be implemented in SPARK (No-cycle lemma). Rust has a similar limitation and requires some work-around to implement structures like graphs (like manually handling indexes).\footnote{https://github.com/nrc/r4cppp/blob/master/graphs/README.md} We also saw that swap primitive is available in SPARK whereas it requires unsafe code in Rust, and the \texttt{in-out} mode allows assigning the exact same path as the one borrowed. In Rust, this feature is available for only one path by moving the argument to the callee, and moving back from the callee the argument as return value of the function. 
	
	Similarly, our intra-procedural analysis has the best approximation possible on aliasing (except for object oriented programming), using the lattice of permissions previously defined that allows tracking the exact path that gets aliased. Rust has a mechanism of restrictions against assignment, mutable borrow, or read-only borrow, which would correspond respectively to \texttt{R/NO}, \texttt{R/W/NO}, or \texttt{W/NO}. For instance, in Rust, when moving a variable, every prefix and extension gets restricted indiscriminately, whereas in SPARK, only paths that access to the aliased data are restricted on their reading or writing permission accordingly. 
	
	Another interesting design choice we made, is to ban borrows or observes inside a block, whereas Rust allows it. Indeed, we do not see use cases for which we need to create several read-only copies of a same deep variable in the body of a subprogram, all the more since this feature would require some new lexical constructs (like keywords or symbols), that cannot be added easily in Ada. Moreover, Ada has a mechanism, called \emph{renaming}, that is used to shorten long paths, hence also rendering useless read-write borrows of a path. 
	
	\section{Conclusion}
	
In this report, we have presented anti-aliasing rules that allow adding pointers to SPARK, a subset of Ada used for static analysis with deductive verification. We showed a systematic analysis that allows a wide range use cases of pointers and dynamic allocation, as showed by the experiments. To the best of our knowledge, this is a novel approach for controlling aliasing introduced by arbitrary pointers in a programming language supported by proof. Our approach does not require user annotations or proof of verification conditions, which makes it much simpler to adopt. Moreover, we provided a mathematical proof of the safety of our analysis, and compared our method to another language providing such analysis, Rust. 

Yet, still some work needs to be done to start using this analysis in commercial applications. Indeed, some constructs of SPARK (specially object orientation) are still not handled efficiently by the rules, and the back-end still does not implement proofs with pointers, which hopefully will be solved during year 2018. Similarly, some Ada libraries have to be changed in order to get them accepted by SPARK (drivers and containers), even though those fixes can be done without major refactoring. We could hopefully expect that future contributors to these anti-aliasing rules will tackle those problems.

Another long term goal would be extending the analysis so that it could handle different features (like Rust), such as automatic reclamation, parallelism, initialization and lifetime checks, which will simplify the existing mechanisms in the GNAT compiler.

	\appendix	
	\pagebreak
	\section*{Appendix}
	\label{sec:Appendix}
	\section{References}
\begin{minipage}{\textwidth}
	\renewcommand{\section}[2]{}%
	\bibliography{report}
	\bibliographystyle{abbrv}
\end{minipage}
	
	\section{Grammar for {\muSPARK}}
	\label{sec:Grammar}

	\subsection*{Basic lexemes}
	
	$$ \begin{array}{lcl}
	\langle digit \rangle & ::= & 0-9 \\
	& & \\
	\langle alpha \rangle & ::= & a-z|A-Z \\
	& & \\
	\langle ident \rangle & ::= & \langle alpha \rangle (\langle alpha \rangle | \langle digit \rangle | \_)^* \\
	& & \\
	
	\langle integer \rangle & ::= & \langle digit \rangle ^+ \\
	& & \\
	\langle name \rangle & ::= & \langle ident \rangle \\
	& | &  \langle name \rangle \texttt{.} \langle ident \rangle\\
	& | &  \langle name \rangle \texttt{.} \texttt{all}\\
	\end{array} $$
	
	\subsection*{Expressions grammar}
	
	$$ \begin{array}{lcl}
	\langle expr \rangle & ::= & \texttt{null} \\
	& | &  \langle integer \rangle  \\
	& | &  \langle name \rangle\texttt{'Access} \\
	& | &  \langle name \rangle\\
	\end{array} $$
	
	\subsection*{Statements grammar}
	
	$$ \begin{array}{lcl}
	\langle instr \rangle & ::= & \langle name \rangle \texttt{:=} \langle expr \rangle \texttt{;} \\
	& | & \langle name \rangle \texttt{:=} ~\texttt{new}~ \langle ident \rangle \texttt{;} \\
	& | & \langle ident \rangle \texttt{(}~ {\langle expr \rangle}^+\hspace{-0.24cm}\texttt{,} ~\texttt{);}\\
	& | &  \texttt{if * then} ~ \langle instr \rangle ~\texttt{else}~\langle instr \rangle ~ \texttt{end if;}\\
	& | &  \texttt{begin}~ \langle  instr \rangle ^+ ~ \texttt{end;}\\
	\end{array} $$
	
	\subsection*{Declarations grammar}

	$$ \begin{array}{lcl}
	& & \\
	\langle file \rangle & ::= & \texttt{procedure}~\langle ident \rangle ~ \texttt{is} ~ \langle decl \rangle ^* ~ \texttt{begin} ~ \langle instr \rangle ~ \texttt{end;} \\
	& & \\
	\langle decl \rangle & ::= &  \texttt{type} ~\langle ident \rangle ~\texttt{is record} ~\langle field \rangle^+~\texttt{end record ;} \\
	& | &\texttt{procedure}~\langle ident \rangle ~ \texttt{(} ~\langle param \rangle ^+ _\texttt{;}~ \texttt{)} ~ \texttt{is} ~ \langle decl \rangle ^* ~ \texttt{begin} ~ \langle instr \rangle ~ \texttt{end;} \\
	& | & \langle ident \rangle : \langle type \rangle\texttt{;} \\
	&& \\
	\langle field \rangle & ::= & \langle ident \rangle ~\texttt{:} ~\langle type \rangle ~\texttt{;}  \\
	&& \\
	
	\langle type \rangle & ::= & \langle ident \rangle  \\
	& | & \texttt{access} ~\langle ident \rangle \\
	&& \\

	\langle param \rangle & ::= & \langle ident \rangle ~\texttt{:} ~\langle mode \rangle ~\langle type \rangle \\
	&& \\
	
	\langle mode \rangle & ::= & \texttt{in} ~|~ \texttt{in-out} ~|~ \texttt{out} \\
	\end{array} $$
	
	\section{Additional conditions on {\muSPARK}}
	
	Some additional conditions are required in order to have valid {\muSPARK}:
	\begin{itemize}
		\item In a procedure call with parameters \texttt{in-out} or \texttt{out}, the parameter should be a $\langle name \rangle$.
		\item It is not possible to write to a formal parameter or any of its fields if it is of mode in.
		\item A declaration cannot shadow another one.
		\item Two fields of a record cannot have the same name.
		\item We consider that all parameters of mode \texttt{in} to procedures are passed by copy. Hence, aliasing between \texttt{in} bound parameters can only occur between formal parameters of deep type.
		\item In record declaration, the record being declared can only be used under access type.
	\end{itemize}

	\section{Typing {\muSPARK}}
	
	The types for {\muSPARK} are defined as the following:
	$$ \begin{array}{lcl}
	\tau& ::= & \texttt{integer} ~|~ R ~|~\texttt{access}~ \tau ~|~ \texttt{nulltype}\\
	\end{array} $$
	where $R$ is a record type. We also write $x:\tau \in R$ if the record type $R$ has a field $x$ of type $\tau$.	We call $\equiv$ the smallest reflexive and symmetric relationship so that $\forall \tau, \texttt{nulltype} \equiv \texttt{access} ~\tau$. We use $\Gamma$ to denote a set of variable, type, and procedure declarations. We write $a \in \Gamma$ if the environment $\Gamma$ contains the declared object $a$. We also use syntactic sugar $\Gamma + d$ to add the object declared by the declaration $d$ to the environment $\Gamma$.
	
	We define deep types by induction. Any $\texttt{access}~\tau$ type is deep, and any record type $R$ in which there is at least one field of deep type is also deep. Any type that is not deep is called shallow.

	\subsection*{Well-formed typing environments}
	
	We call well-formed any typing environment $\Gamma$ that is built constructively from a well-typed program, starting with the empty environment. We assume that in the following typing rules, the environments are well-formed.
	
	\subsection*{Typing rules for names}
	
	We write $\Gamma \vdash_\textnormal{\small n} n : \tau$ if ``in the environment $\Gamma$, the name $n$ is well-typed and has type $\tau$''.
	
	\infrule[T-readIdent]
	{x : \tau \in \Gamma} 
	{\Gamma  \vdash_\textnormal{\small n} {x : \tau }}

	\infrule[T-readField]
	{\Gamma \vdash_\textnormal{\small n} n : R \andalso a : \tau \in R} 
	{\Gamma  \vdash_\textnormal{\small n} {n.a : \tau }}

	\infrule[T-readDeref]
	{\Gamma \vdash_\textnormal{\small n} n : \texttt{access} ~\tau } 
	{\Gamma  \vdash_\textnormal{\small n} {n.\texttt{all} : \tau }}
	
	\subsection*{Typing rules for expressions}
	
	We write $\Gamma \vdash_\textnormal{\small e} e : \tau$ if ``in the environment $\Gamma$, the expression $e$ is well-typed and has type $\tau$''.
	
	\infrule[T-nullValue]
	{} 
	{\Gamma \vdash_\textnormal{\small e} \texttt{null} : \texttt{nulltype}}

	\infrule[T-literal]
	{e ~\textnormal{integer literal}} 
	{\Gamma  \vdash_\textnormal{\small e} { e : \texttt{integer} }}

	\infrule[T-takeAccess]
	{\Gamma \vdash_\textnormal{\small n} n : \tau} 
	{\Gamma  \vdash_\textnormal{\small e} {n\texttt{'Access} : \texttt{access} ~\tau }}

	\infrule[T-name]
	{\Gamma \vdash_\textnormal{\small n} n : \tau} 
	{\Gamma  \vdash_\textnormal{\small e} {n : \tau }}
	
	\subsection*{Typing rules for statements}
	We write $\Gamma \vdash_\textnormal{\small s} s$ if ``in the environment $\Gamma$, the statement $s$ is well-typed''.

	\infrule[T-assignExpr]
	{\Gamma  \vdash_\textnormal{\small n} {x:\tau'} \andalso \Gamma \vdash_\textnormal{\small e} e : \tau \andalso \tau \equiv \tau'} 
	{\Gamma  \vdash_\textnormal{\small s} {x \texttt{:=}~ e}}

	\infrule[T-assignNew]
	{\Gamma  \vdash_\textnormal{\small n} {x:\tau'} \andalso \tau \in \Gamma \andalso \texttt{access}~ \tau \equiv \tau' } 
	{\Gamma  \vdash_\textnormal{\small s} {x \texttt{:=}~ \texttt{new}~ \tau}}
	
	\infrule[T-procedureCall]
	{\texttt{procedure} ~P(x_1:m_1 ~\tau_1',~...~,x_n:m_n ~\tau_n') \in \Gamma \\  \forall i \in \{1..k\}, m_i = \texttt{in} \\ \forall i \leq k, ~(\Gamma \vdash_\textnormal{\small e} e_i : \tau_i) \wedge (\tau_i \equiv \tau_i') \andalso \forall i > k, ~(\Gamma \vdash_\textnormal{\small n} n_i : \tau_i) \wedge (\tau_i \equiv \tau_i')} 
	{\Gamma  \vdash_\textnormal{\small s} {P(e_1, ~...~, e_k, n_{k+1}~... ~, n_n)}}
	
	\infrule[T-block]
	{\forall j,~ \Gamma \vdash_\textnormal{\small s} i_j} 
	{\Gamma  \vdash_\textnormal{\small s} {\texttt{begin}~ i_1 ~...~ i_n~ \texttt{end}}}
	
	\infrule[T-ifCondition]
	{\forall j,~ \Gamma \vdash_\textnormal{\small s} i_j} 
	{\Gamma  \vdash_\textnormal{\small s} {\texttt{if * then} ~ i_1 ~ \texttt{else} ~ i_2~ \texttt{end if}}}
	
	\subsection*{Typing rules for declarations}
	We write $\Gamma \vdash_\textnormal{\small t} \tau$ if ``in the environment $\Gamma$, the type $\tau$ is well formed''.

	\infrule[T-baseType]
	{} 
	{\Gamma  \vdash_\textnormal{\small t} {\texttt{integer}}}
	
	\infrule[T-accessType]
	{\Gamma  \vdash_\textnormal{\small t} \tau} 
	{\Gamma \vdash_\textnormal{\small t} {\texttt{access}~\tau}}
	
	\infrule[T-recordType]
	{R \in \Gamma} 
	{\Gamma \vdash_\textnormal{\small t} {R}}
	
	We write $\Gamma \vdash_\textnormal{\small d} d$ if ``in the environment $\Gamma$, the declaration $d$ is well formed''. Note that usage of global variables is forbidden inside procedures.

	\infrule[T-uninitDecl]
	{\Gamma \vdash_\textnormal{\small t} \tau} 
	{\Gamma \vdash_\textnormal{\small d} x : \tau}
	
	\infrule[T-procedureDecl]
	{\forall j, \Gamma \vdash_\textnormal{\small t} \tau_i' \andalso \Gamma' := \{\tau \in \Gamma, \texttt{procedure}~P'\in \Gamma\} + P + x_1:\tau_1' + ... + x_n : \tau_n' \\
	\forall k, \Gamma' + d_1 + ... + d_{k-1} \vdash_\textnormal{\small d} d_k \andalso \Gamma' + d_1 + ... + d_m \vdash_\textnormal{\small s} i} 
	{\Gamma \vdash_\textnormal{\small d} \texttt{procedure} ~P(x_1:\tau_1', ... , x_n : \tau_n') ~\texttt{is} ~d_1, ..., d_m ~\texttt{begin}~ i~ \texttt{end}}
	
	\infrule[T-recordDecl]
	{\forall i, \Gamma + R \vdash_\textnormal{\small t} \tau_i} 
	{\Gamma \vdash_\textnormal{\small d} {\texttt{type} ~R~ \texttt{is record}~ x_1:\tau_1, ..., x_n:\tau_n} ~\texttt{end record;}}
	
\section{Permission rules for µSPARk}
\label{sec:Permission Rules for MuSPARK}

\subsection*{Definitions}

\subsubsection*{Permission environments}

Permission environment $\Phi$ are mappings from variable declarations to permissions trees. As shown by the permission rules, permission environments evolve through the execution of the program. Note that in case of a recursive record type, the permission tree is infinite.

$$ \begin{array}{lcl}
P & ::= & Integer(Permission) \\
& | & Record(Permission, Fields \rightarrow P) \\
& | & Access(Permission, P) \\
& & \\
\end{array} $$

\subsubsection*{Normalized permission environment}
We say that a permission environment is normalized when each node of each tree has a permission that is more permissive than the most restrictive permission of its children.

\subsubsection*{PermRelease operator}
\label{sec:PermRelease}
We define a normalization operator called $\textit{PermRelease}$ that takes a permission environment as input, and produces another one with permissions of each node changed using the following definition:

$$
 \begin{array}{lcl}
\textit{PermRelease} ~(Integer(\kappa))  & = & Integer(\kappa) \\
\textit{PermRelease} ~(Record(\kappa, f_i \rightarrow F_i))  & = &  Record(\kappa \vee (\bigwedge_{i} \textit{PermRelease}~(F_i).Permission), \\
&&f_i \rightarrow \textit{PermRelease} ~(F_i))\\
\textit{PermRelease} ~(Access(\kappa, D))  & = & Access(\kappa \vee \textit{PermRelease}~ (D).Permission, \textit{PermRelease} ~(D)) \\
& & \\
\end{array} 
$$

This produced environment is trivially normalized.

\subsubsection*{Fusion operator}
We define an operator called $\textit{Fusion}$ that takes two (or more) permission environments as input, and outputs another one where for each node, the permission is updated by taking the glb of every identical node from the merged environments.

\subsubsection*{Pfresh(tau,kappa)}

We also define the function $\textit{Pfresh}(\tau,\kappa)$ that takes a type $\tau$ and a permission $\kappa$ as input and produces an adequate permission tree.

More specifically, \begin{itemize}
	\item $\textit{Pfresh}(\texttt{integer}, \kappa)$ is equivalent to $Integer(\kappa)$.
	\item $\textit{Pfresh}(\texttt{access}~ \tau,\kappa)$ is equivalent to $Access(\kappa, \textit{Pfresh}(\tau,\kappa))$.
	\item $\textit{Pfresh}(R,\kappa)$ is equivalent to $Record(\kappa, \lambda x:\tau \in R. \textit{Pfresh}(\tau,\kappa))$.
\end{itemize}

\subsubsection*{Syntactic sugar}

We allow modifying or accessing permission and memory trees at a given paths in the environments. For instance, $x.a.all.y\mapsto \textit{Pfresh}(\tau)$ will modify the node $x.a.all.y$ in the tree of x, by replacing its subtree with a new (fresh) tree.

\subsection*{Permission rules for names}
We annotate each permission rule with $_n$, $_e$, $_s$ or $_d$ for readability purposes.

\subsubsection*{Move}
\infrule[P-M-ident]
{} 
{\Phi.~ident \xrightarrow{Move}_n \Phi[n\mapsto\texttt{W}]}

\infrule[P-M-field]
{\Phi.~n \xrightarrow{Move}_n \Phi' } 
{\Phi.~n.a \xrightarrow{Move}_n \Phi'[n.a\mapsto\texttt{W}] }

\infrule[P-M-deref]
{\Phi.~n \xrightarrow{Move}_n \Phi' } 
{\Phi.~n.all \xrightarrow{Move}_n \Phi'[n.all\mapsto\texttt{W}] }

\infrule[P-SM-ident]
{} 
{\Phi.~ident \xrightarrow{Access}_n \Phi[n\mapsto\texttt{NO}]}

\infrule[P-SM-field]
{\Phi.~n \xrightarrow{Access}_n \Phi' } 
{\Phi.~n.a \xrightarrow{Access}_n \Phi'[n.a\mapsto\texttt{NO}] }

\infrule[P-SM-deref]
{\Phi.~n \xrightarrow{Move}_n \Phi' } 
{\Phi.~n.all \xrightarrow{Access}_n \Phi'[n.all\mapsto\texttt{NO}] }

\subsubsection*{Borrow}
\infrule[P-B-identDeep]
{\Gamma \vdash_n n : \tau \andalso \tau ~\textnormal{deep}} 
{\Phi.~ident \xrightarrow{Borrow}_n \Phi[n\mapsto\texttt{NO}]}

\infrule[P-B-fieldDeep]
{\Gamma \vdash_n n.a : \tau \andalso \tau ~\textnormal{deep} \andalso \Phi.n \xrightarrow{Borrow}_n \Phi'} 
{\Phi.~n.a \xrightarrow{Borrow}_n \Phi'[n.a\mapsto\texttt{NO}] }

\infrule[P-B-derefDeep]
{ \Gamma \vdash_n n.all : \tau \andalso \tau ~\textnormal{deep} \andalso \Phi.~n \xrightarrow{Borrow}_n \Phi'} 
{\Phi.~n.all \xrightarrow{Borrow}_n \Phi'[n.all\mapsto\texttt{NO}] }

\infrule[P-B-identShallow]
{\Gamma \vdash_n n : \tau \andalso \tau ~\textnormal{shallow}} 
{\Phi.~ident \xrightarrow{Borrow}_n \Phi[n\mapsto\texttt{R}]}

\infrule[P-B-fieldShallow]
{\Gamma \vdash_n n.a : \tau \andalso \tau ~\textnormal{shallow} \andalso \Phi.~n \xrightarrow{Borrow}_n \Phi' } 
{\Phi.~n.a \xrightarrow{Borrow}_n \Phi'[n.a\mapsto\texttt{R}] }

\infrule[P-B-derefShallow]
{\Gamma \vdash_n n.all : \tau  \andalso \tau ~\textnormal{shallow} \andalso \Phi.~n \xrightarrow{Borrow}_n \Phi'} 
{\Phi.~n.all \xrightarrow{Borrow}_n \Phi'[n.all\mapsto\texttt{R}] }

\infrule[P-B-entryPointInOut]
{\Phi(n) = \texttt{RW}  \andalso \Phi.~n \xrightarrow{Borrow}_n \Phi'\\ 
	\Phi'' = \Phi'[\forall n' ~ \textnormal{strict deep extension of} ~ n, n' \mapsto \texttt{NO},\\ \forall n' ~ \textnormal{strict shallow extension of} ~ n, n' \mapsto \texttt{R}]}
{\Phi.~n \xrightarrow{BorrowInOut}_e \Phi''}

\infrule[P-B-entryPointOut]
{\Phi(n) = \texttt{W,RW} \andalso \Phi.~n \xrightarrow{Borrow}_n \Phi'\\ 
	\Phi'' = \Phi'[\forall n' ~ \textnormal{strict deep extension of} ~ n, n' \mapsto \texttt{NO},\\ \forall n' ~ \textnormal{strict shallow extension of} ~ n, n' \mapsto \texttt{R}]}
{\Phi.~n \xrightarrow{BorrowOut}_n \Phi''}
\subsubsection*{Observing}
\infrule[P-O-ident]
{\Gamma \vdash_n n : \tau \andalso \tau ~\textnormal{deep}} 
{\Phi.~ident \xrightarrow{ObserveName}_n \Phi[n\mapsto\texttt{R}]}

\infrule[P-O-field]
{\Gamma \vdash_n n.a : \tau \andalso \tau ~\textnormal{deep} \andalso \Phi.~n \xrightarrow{ObserveName}_n \Phi' } 
{\Phi.~n.a \xrightarrow{ObserveName}_n \Phi'[n.a\mapsto\texttt{R}] }

\infrule[P-O-deref]
{\Gamma \vdash_n n.all : \tau \andalso \tau ~\textnormal{deep} \andalso \Phi.~n \xrightarrow{ObserveName}_n \Phi'} 
{\Phi.~n.all \xrightarrow{ObserveName}_n \Phi'[n.all\mapsto\texttt{R}] }

\infrule[P-O-entryPoint]
{\Phi(n) = \texttt{R,RW} \andalso \Phi.~n \xrightarrow{ObserveName}_n \Phi'\\ 
	\Phi'' = \Phi'[\forall n' ~ \textnormal{strict extension of} ~ n, n' \mapsto \texttt{R}]}
{\Phi.~n \xrightarrow{Observe}_n \Phi''}

\subsection*{Permission rules for expressions}

\subsubsection*{Borrow (in or in-out mode)}

Note that it is not possible to borrow an integer literal.

\infrule[P-B-nullValue]
{} 
{\Phi.~\texttt{null} \xrightarrow{BorrowInOut}_e \Phi}

\infrule[P-B-takeAccess]
{\Phi.~n \xrightarrow{BorrowInOut}_n \Phi'}
{\Phi.~n \texttt{'Access} \xrightarrow{BorrowInOut}_e \Phi'}

\infrule[P-B-name]
{\Phi.~n \xrightarrow{BorrowInOut}_n \Phi'}
{\Phi.~n \xrightarrow{BorrowInOut}_e \Phi'}

\subsubsection*{Borrow (out mode)}

Only names can be borrowed with mode \texttt{out}.

\infrule[P-B-name-Out]
{\Phi.~n \xrightarrow{BorrowOut}_n \Phi'}
{\Phi.~n \xrightarrow{BorrowOut}_e \Phi'}

\subsubsection*{Observing}

\infrule[P-O-nullValue]
{} 
{\Phi.~\texttt{null} \xrightarrow{Observe}_e \Phi}

\infrule[P-O-literal]
{e ~\textnormal{integer literal}} 
{\Phi.~e \xrightarrow{Observe}_e \Phi}

\infrule[P-O-takeAccess]
{\Phi.~n \xrightarrow{ObserveName}_n \Phi'}
{\Phi.~n \texttt{'Access} \xrightarrow{Observe}_e \Phi''}

\infrule[P-O-name]
{\Phi.~n \xrightarrow{ObserveName}_n \Phi'}
{\Phi.~n \xrightarrow{Observe}_e \Phi''}

\subsection*{Permission rules for statements}

\infrule[P-assignNull]
{x:\tau \in \Gamma \andalso \Phi(x) = \texttt{W,RW}} 
{\Phi.~x \texttt{:= null} \xrightarrow{}_s \textit{PermRelease}(\Phi[x\mapsto \textit{Pfresh}(\tau,\texttt{RW})])}

\infrule[P-assignLiteral]
{x:\tau\in\Gamma \andalso e~ \textnormal{literal} \andalso \Phi(x) = \texttt{W,RW}} 
{\Phi.~x \texttt{:=}~ e \xrightarrow{}_s \textit{PermRelease}(\Phi[x\mapsto \textit{Pfresh}(\tau,\texttt{RW})])}

\infrule[P-assignDeepName]
{x:\tau\in\Gamma \andalso n~ \textnormal{name} \andalso \Phi(n) = \texttt{RW} \andalso \Phi.~n \xrightarrow{Move}_n \Phi' \\ \tau ~\textnormal{deep} \andalso \Phi'' = \Phi'[\forall n' ~ \textnormal{strict extension of} ~n~\textnormal{with more \texttt{.all}}, n' \mapsto \texttt{NO} \\ \forall n' ~ \textnormal{strict deep extension of} ~n ~\textnormal{with same number of \texttt{.all}}, n' \mapsto \texttt{W}  \\ \forall n' ~ \textnormal{strict shallow extension of} ~n ~\textnormal{with same number of \texttt{.all}}, n' \mapsto \texttt{RW} ] \\  \Phi''(x) = \texttt{W,RW}} 
{\Phi.~x \texttt{:=}~ n \xrightarrow{}_s \textit{PermRelease}(\Phi''[x\mapsto \textit{Pfresh}(\tau,\texttt{RW}))])}

\infrule[P-assignShallowName]
{x:\tau\in\Gamma \andalso n~ \textnormal{name} \\   \Phi(n) = \texttt{R} \andalso \Phi(x) = \texttt{W,RW} \andalso \tau ~\textnormal{shallow}} 
{\Phi.~x \texttt{:=}~ n \xrightarrow{}_s \textit{PermRelease}(\Phi[x\mapsto \textit{Pfresh}(\tau,\texttt{RW}))])}

\infrule[P-assignAccess]
{x:\tau\in\Gamma \andalso n~ \textnormal{name} \andalso \Phi(n) = \texttt{RW} \andalso \Phi.~n \xrightarrow{Access}_n \Phi' \\ 
	\Phi'' = \Phi'[\forall n' ~ \textnormal{strict extension of} ~n, n' \mapsto \texttt{NO}] \\  \Phi''(x) = \texttt{W,RW}} 
{\Phi.~x \texttt{:=}~ n\texttt{'Access} \xrightarrow{}_s \textit{PermRelease}(\Phi''[x\mapsto \textit{Pfresh}(\tau,\texttt{RW})])}

\infrule[P-assignNew]
{\Phi(x) = \texttt{W,RW}} 
{\Phi.~x \texttt{:=}~ \texttt{new}~\tau \xrightarrow{}_s\textit{PermRelease}(\Phi[x\mapsto \textit{Pfresh}(\tau,\texttt{RW})])}

\infrule[P-call]
{\texttt{procedure} ~P(x_1:\tau_1,~...~,x_n:\tau_n) ~\textnormal{with}~\\ x_1...x_a~ \textnormal{observed},\\ ~ x_{a+1}...x_b~ \textnormal{borrowed with mode \texttt{in}}, \\ ~ x_{b+1}...x_c~ \textnormal{borrowed with mode \texttt{in-out}}, \\~ x_{c+1}...x_n~ \textnormal{borrowed with mode \texttt{out}}, \\ 
	\Phi_1 = \Phi \\ 
	\forall~ 0 < i \leq a, \Phi_i.~e_i \xrightarrow{Observe}_e \Phi_{i+1} \\ 
	\forall~ a < i \leq b, \Phi_i.~e_i \xrightarrow{BorrowInOut}_e \Phi_{i+1} \\
	\forall~ b < i \leq c, \Phi_i.~e_i \xrightarrow{BorrowInOut}_e \Phi_{i+1} \\   
	\forall~ c < i \leq n, \Phi_i.~e_i \xrightarrow{BorrowOut}_e \Phi_{i+1} \\
	\Phi'=\textit{PermRelease}(\Phi[\forall~b < i \leq n: e_i \mapsto  \textit{Pfresh}(\tau_i,\texttt{RW})])
}
{\Phi.~P(e_1, ~...~, e_n) \xrightarrow{}_s \Phi'}

\infrule[P-block]
{\forall j>0,~ \Phi_j .~ i_j \xrightarrow{}_s \Phi_{j+1}} 
{\Phi_1. ~{\texttt{begin}~ i_1 ~...~ i_n~ \texttt{end}}  \xrightarrow{}_s \Phi_{n+1}}

\infrule[P-ifCondition]
{	\Phi . ~i_1 \xrightarrow{}_s \Phi' \andalso
	\Phi .~ i_2 \xrightarrow{}_s \Phi''
}
{\Phi .~\texttt{if * then} ~ i_1 ~ \texttt{else} ~ i_2~ \texttt{end if} \xrightarrow{}_s  \textit{PermRelease}(\textit{Fusion}(\Phi', \Phi'')) }

\subsection*{Permission rules for declarations}

\infrule[P-uninitDecl]
{} 
{\Phi.~ x : \tau \texttt{;} \xrightarrow{}_d \Phi[x\mapsto \textit{Pfresh}(\tau,\texttt{W})]}

\infrule[P-procedureDecl]
{\texttt{procedure} ~P(x_1,~...~,x_n) ~\textnormal{with}~\\ x_1...x_a~ \textnormal{observed},\\ ~ x_{a+1}...x_b~ \textnormal{borrowed with mode \texttt{in} or \texttt{in-out}}, \\~ x_{b+1}...x_n~ \textnormal{borrowed with mode \texttt{out}}, 
	\\
	\Phi'=\{\forall~0<i\leq a, x_i \mapsto \textit{Pfresh}(\tau_i,\texttt{R}), \\ 
	\forall~a<i\leq b, x_i \mapsto \textit{Pfresh}(\tau_i,\texttt{RW}),\\
	\forall~b<i\leq n, x_i \mapsto \textit{Pfresh}(\tau_i,\texttt{W})\}\\
	\Phi_1 = \Phi' \\	
	\forall k, \Phi_k . ~d_k \xrightarrow{}_d \Phi_{k+1}
	\\
	\Phi_{m+1} . ~i \xrightarrow{}_d \Phi'' \\ 
	\forall i \in \{a+1..n\}, \Phi''(x_i) = \texttt{RW}} 
{\Phi .~\texttt{procedure} ~P(x_1, ... , x_n) ~\texttt{is} ~d_1, ..., d_m ~\texttt{begin}~ i~ \texttt{end} \xrightarrow{}_d \Phi''}

\section{Theorems and proofs}
\label{sec:Proof}
We want to show that at each point of the program, the following lemmas and theorem hold.

\setcounter{lemma}{0}
\subsection*{Lemmas}

\begin{lemma}\textnormal{(Normalization)\textbf{.}}
	Permission environment is normalized. 
\end{lemma}
\begin{proof}
	Straightforward. Every rule of the semantics calls the normalization operator \textit{PermRelease}. The rules that do not call it are \textsc{(P-block)} (trivial case), and declaration rules (they create only new trees with one node).
\end{proof}

\begin{lemma}\textnormal{(Coherence)\textbf{.}}
	For every node $n$ in the memory environment, it is possible to find an associated node $n'$ in the permission environment, such that $n$ and $n'$ designate the same path.
\end{lemma}
\begin{proof}
	Every permission rule manipulates permission trees that have the same constructs as memory trees, except for access types that cannot be null. Hence every memory tree can be obtained from a permission tree by cutting nodes at access nodes when the pointer has \texttt{null} value.
\end{proof}

\begin{lemma}\textnormal{(Readability)\textbf{.}} 
	If a node has permission \texttt{R} (resp. \texttt{RW}), then all its children have permission \texttt{R} (resp. \texttt{RW}) at each step of the semantics.
\end{lemma}

\begin{proof}
	By induction on the semantics: 
	\begin{enumerate}
		\item \textsc{(E-assignNull)}, \textsc{(E-assignLiteral)}: $\textit{Pfresh}(\tau, \texttt{RW})$ guarantees it for the assigned path.
		\item \textsc{(E-assignName)}, \textsc{(E-assignAccess)}, \textsc{(E-assignNew)}: same thing for the assigned path. For the moved path, every prefix gets its permission set to either \texttt{W} or \texttt{NO} by the \textsc{(P-M-ident)}, \textsc{(P-M-field)}, \textsc{(P-M-deref)} or \textsc{(P-SM-ident)}, \textsc{(P-SM-field)}, \textsc{(P-SM-deref)} rules. The extensions are handled by \textsc{(P-assignDeepName)}, \textsc{(P-assignShallowName)}, \textsc{(P-assignAccess)}, and \textsc{(P-assignNew)}.
		\item \textsc{(E-call)}: we create such nodes by observing their paths. The rules \textsc{(P-O-nullValue)}, \textsc{(P-O-litteral)}, \textsc{(P-O-takeAccess)}, \textsc{(P-O-name)} guarantee that if any node is set to permission \texttt{R}, then any extension of it is also set to \texttt{R}. After the callee returned, the proof is identical to assignments.
		\item \textsc{(E-block)}, \textsc{(E-ifConditionTrue)}, \textsc{(E-ifConditionFalse)}, \textsc{(E-uninitDecl)}, \textsc{(E-procedureDecl)}: trivial
	\end{enumerate}
\end{proof}

\begin{lemma}\textnormal{(No-cycle)\textbf{.}} 
	A node cannot have the same memory cell as any of its descendants (hence memory trees are finite).
\end{lemma}

\begin{proof}
	By induction on the semantics: 
	\begin{enumerate}
		\item \textsc{(E-assignNull)}, \textsc{(E-assignLiteral)}, \textsc{(E-assignNew)}: given that those rules only cut memory trees, there are no new indirections created.
		\item \textsc{(E-assignName)}: such a cycle could only be created when applying $Assign$ on a node $Access(C1, \_)$ which is a descendant of the node being moved $Access(\_, P)$. In such a case, this would contradict hypothesis of the rule \textsc{(P-assignDeepName)}, given that we set  descendants of the moved node (ie strict extensions) to \texttt{NO}, before checking \texttt{RW} for the node to assign.
		\item \textsc{(E-assignAccess)}: such a cycle could only be created if the assigned $Access$ node is a descendant of the subtree we are taking address of. In such a case, this would contradict hypothesis of the rule \textsc{(P-assignAccess)}, given that we set  descendants of the moved node (ie strict extensions) to \texttt{NO}, before checking \texttt{RW} for the node to assign.
		\item \textsc{(E-call)}: borrows and observes do not modify addresses before transferring to callee. After returning, we assign every parameter. Hence identical as \textsc{(E-assign)} on each \texttt{in-out} or \texttt{out} parameter. For every \texttt{in} parameter, \textit{SetFromExpr} is equivalent to either assigning $n.all$ or $n$.
		\item \textsc{(E-block)}, \textsc{(E-procedureDecl)}: trivial by applying the induction hypothesis successively to each statement of the block.
		\item \textsc{(E-ifConditionTrue)}, \textsc{(E-ifConditionFalse)}: the $Fusion$ operator does not change memory places and pointers, only permissions. 
		\item \textsc{(E-uninitDecl)}: every node is assigned to a new memory area (and pointers to null). Hence it is impossible to point to an existing memory area.
	\end{enumerate} 
	
\end{proof}

\subsection*{Non aliasing}

\setcounter{theorem}{0}
\begin{theorem}\textnormal{(No-aliasing)\textbf{.}}
	For every set of nodes $S$ in memory environment such that their memory cell is identical, consider their associated paths. Then, if there is one path that can be written, then all other paths in $S$ can neither be written nor read.
\end{theorem}
\begin{preuve}
To show that each theorem holds at each point of the program, we reason by induction on the semantics: 

\begin{enumerate}
	\item \textsc{} \infrule[E-assignNull]
	{} 
	{\Upsilon.~x \texttt{:= null} \xRightarrow{}_s \Upsilon[x.all\mapsto \texttt{Null}]}
	\setcounter{XxmpX}{0}
		Let us take any set $S$ in the environments $\Phi', \Upsilon'$ after executing one step of the semantics. Given that we do not modify addresses during assignment, this set $S$ also verifies, by induction hypothesis, non-aliasing before executing that step, in the environments $\Phi, \Upsilon$. There is in $S$ at most one ancestor of $x$ \textnormal{(No-cycle on $\Phi$)}. If there are none, there is nothing to prove. If there is one (that we name $y$), let us consider the two cases of non-aliasing in $\Phi, \Upsilon$, before executing that step.
		\begin{itemize}
			\item If $y$ is the path that can be written (there is only one by induction hypothesis), then using the fact paths not related to $x$ do not get their permissions changed, we show that all others paths cannot be written or read after executing one step of the semantics. 
			\item By (Readability), $y$ cannot be \texttt{R} in $\Phi$. Hence it can only be \texttt{NO}, which leads to a contradiction. This means that if $y$ can be read, then it can also be written.
			\item If $y$ cannot be written nor read, then its permission is \texttt{NO} before executing the semantics. By normalization of $\Phi$, we can say that the glb of its children is also \texttt{NO} (and recursively). Which contradicts \texttt{RW} permission for $x$.
		\end{itemize} 
	\item \textsc{}
	\infrule[E-assignLiteral]
	{x:\tau\in\Gamma \andalso e~ \textnormal{literal}} 
	{\Upsilon.~x \texttt{:=}~ e \xRightarrow{}_s \Upsilon[x.value \mapsto e]}
	Exactly the same as \textsc{(E-assignNull)}.

	\item \textsc{} \infrule[E-assignName]
	{n~ \textnormal{name}} 
	{\Upsilon.~x \texttt{:=}~ n \xRightarrow{}_s\Upsilon[Assign(x, \Upsilon(n))]}
	
	Let us take any set $S$ in the environments $\Upsilon',\Phi'$ after executing one step of the semantics. Nodes can only get aliased between the first \textit{Access} node descendant of assigned and moved node (because values are recursively copied, up to the first \textit{Access} node encountered in which the pointer is copied, which creates an alias). Like \textsc{(E-assignNull)}, we can consider only the case when the moved or assigned path is the one that can be written. 
	
	Executing one step of the semantics creates an alias in $x$ (the assigned path) for every node of any subtree rooted under an $Access$ node that is a direct descendant of $n$. However, the rule \textsc{(P-assignDeepName)} says that when encountering an $Access$ node, all the further descendants are set to \texttt{NO}, hence their paths cannot be written. The rules \textsc{(P-M-ident)}, \textsc{(P-M-field)}, \textsc{(P-M-deref)} guarantee also that any extension of $n$ cannot be read anymore. This solves the case for elements of $S$ that are descendant of $x$ with different number of \texttt{.all}. 
	
	For elements of $S$ that are descendant of $x$ with same number of \texttt{.all}, no alias is created during assignment, hence induction hypothesis can be applied directly. 
	
	Hence, the assigned path gets \texttt{RW} permission (as well as any extensions), but any aliased subtree of this node cannot be written nor read in the moved tree, hence only one aliased path be written, and all others can neither be written nor read.

\item \textsc{}
\infrule[E-assignAccess]
{n~ \textnormal{name}} 
{\Upsilon,\Phi.~x \texttt{:=}~ n\texttt{'Access} \xRightarrow{}_s \Upsilon[x.all\mapsto \Upsilon(n)]}
	
	Let us take any set $S$ in the environments $\Phi', \Upsilon'$ after executing one step of the semantics. Nodes can only get aliased from the subtree we take address and the one that is assigned using \texttt{'Access}. Like \textsc{(E-assignNull)}, we can only consider exclusive write aliasing, as well as consider the case when the moved or assigned path is the one that can be written. 
	
	Executing one step of the semantics creates an alias of $x.all$ (the assigned path) with the subtree rooted at $n$. However, the rules \textsc{(P-assignAccess)}, \textsc{(P-SM-ident)}, \textsc{(P-SM-field)}, \textsc{(P-SM-deref)} ensure that all the subtree as well as its parents up to the first Access node encountered are set to \texttt{NO}, hence any of the aliased path (extensions of $n$) cannot be written. The rules \textsc{(P-M-ident)}, \textsc{(P-M-field)}, \textsc{(P-M-deref)} that get applied when applying \textsc{(P-SM-deref)} also guarantee that those aliased paths cannot be read.
	
	Hence, the assigned path gets \texttt{RW} permission (as well as any extension), but any aliased subtree of this node cannot be written or read in the moved tree, hence only one aliased path can be written, and all others can neither be written nor read.

\item \textsc{} \infrule[E-assignNew]
{} 
{\Upsilon.~x \texttt{:=}~ \texttt{new}~\tau \xRightarrow{}_s\Upsilon[x.all\mapsto \textit{fresh}(\tau)]}
\setcounter{XxmpX}{0} 
Exactly the same as \textsc{(E-assignNull)}, given that we allocate fresh memory area not aliased with anything existing before.

\item \textsc{}
\infrule[E-call]
{e_{1}..e_a~ \textnormal{with mode \texttt{in}}, \\ e_{a+1}...e_b~ \textnormal{with mode \texttt{in-out}}, \\~ e_{b+1}...e_n~ \textnormal{with mode \texttt{out}}, \\ \Upsilon'=\{\forall~ 0 < i \leq n, x_i \mapsto \textit{GetFromExpr}_{\Upsilon}(e_i) \} \\
	\Upsilon'.~\texttt{procedure} ~P(x_1,~...~,x_n) \xRightarrow{}_d \Upsilon'' \\
	\Upsilon'''=\Upsilon[\forall~ 0 < i \leq a, \textit{SetFromExpr}_{\Upsilon''}(e_i,x_i)\\
	\forall~ a < i \leq b, Assign(e_i, \Upsilon''(x_i)) \\
	\forall~ b < i \leq n, Assign(e_i, \Upsilon''(x_i))] 
}
{\Upsilon.~P(e_1, ~...~, e_n) \xRightarrow{}_s \Upsilon'''}

The case for procedure calls is particular: indeed we apply our inductive hypothesis after modifying environments, hence we have to guarantee that the invariants not only hold at the end of the rule, but also at the moment we transfer the control flow to the callee (ie for $\Upsilon'$ in \textsc{(P-call)} and $\Phi'$ in \textsc{(P-procedureDecl)}).

Let us consider such a set $S$ in $\Upsilon'$ (at the moment transferring to the callee). Suppose that one element of $S$ can be written. Then the rule \textsc{(P-procedureDecl)} guarantee that this element is borrowed (it is a borrow of a subtree $T$ of $\Upsilon$). Hence the rules \textsc{(P-B-name)}, \textsc{(P-B-name-Out)}, \textsc{(P-B-takeAccess)}, \textsc{(P-B-nullValue)}, guarantee that the borrowed subtree has permission set to \texttt{NO} if deep, and \texttt{R} if shallow. Note that the borrowed subtree cannot have been observed, given that we borrow after observing, hence we require \texttt{RW} permission on an argument that has been set to \texttt{R} by observation. Thus, it is impossible for $T$ to be observed, hence creating another element of $S$ that could be readable. For the same reason, the hypothesis of the rules \textsc{(P-B-entryPointInOut)} and  \textsc{(P-B-entryPointOut)} guarantee that the argument cannot be borrowed twice. Hence if one path can be written, all others cannot be neither written nor read.
\item \textsc{} 
\infrule[E-block]
{\forall j>0,~ \Upsilon_j .~ i_j \xRightarrow{}_s \Upsilon_{j+1}} 
{\Upsilon_1. ~{\texttt{begin}~ i_1 ~...~ i_n~ \texttt{end}}  \xRightarrow{}_s \Upsilon_{n+1}}
Trivial by applying the induction hypothesis successively to each statement of the block.

\item \textsc{} 

\infrule[E-ifConditionTrue]
{\Upsilon .~ i_1 \xRightarrow{}_s \Upsilon' } 
{\Upsilon .~\texttt{if * then} ~ i_1 ~ \texttt{else} ~ i_2~ \texttt{end if} \xRightarrow{}_s \Upsilon'}
	Let us consider such a set $S$. At the end of executing the first block, they are non-aliased, at the end of the second block, they are also non-aliased. Hence applying \textit{Fusion} that only takes the $glb$ cannot give more permissions to a path than it has before. Hence aliasing cannot occur.
\item \textsc{}

\infrule[E-ifConditionFalse]
{\Upsilon .~ i_2 \xRightarrow{}_s \Upsilon' } 
{\Upsilon .~\texttt{if * then} ~ i_1~ \texttt{else} ~ i_2~ \texttt{end if} \xRightarrow{}_s \Upsilon'}

Exactly the same as \textsc{(E-ifConditionTrue)}.

\item \textsc{} 
\infrule[E-uninitDecl]
{x:\tau\in\Gamma} 
{\Upsilon.~ x : \tau \xRightarrow{}_d \Upsilon[x\mapsto \textit{fresh}(\tau)]}
	It cannot point to existing memory area, hence cannot alias.
	
\item \textsc{} \infrule[E-procedureDecl]
{
	\Upsilon_1 = \Upsilon \andalso	\forall k, \Upsilon_k . ~d_k \xRightarrow{}_d \Upsilon_{k+1} \andalso
	\Upsilon_{m+1} . ~i \xRightarrow{}_d \Upsilon' }
{\Upsilon .~\texttt{procedure} ~P(x_1, ... , x_n) ~\texttt{is} ~d_1, ..., d_m ~\texttt{begin}~ i~ \texttt{end} \xRightarrow{}_d \Upsilon'} 
	Like \textsc{(E-block)}, we use inductive hypothesis to each declaration and then each instruction of the body.
\end{enumerate} 

\end{preuve}

\end{document}